\PassOptionsToPackage{final}{graphicx}
\PassOptionsToPackage{nosumlimits,nonamelimits}{amsmath}

\documentclass[conference,10pt,final]{IEEEtran}

\usepackage{etoolbox} %

\usepackage{ifdraft}

\makeatletter
\@draftfalse
\makeatother

\usepackage[colorlinks,linkcolor={blue},citecolor={blue},urlcolor={red},breaklinks=true,final]{hyperref}

\usepackage[english]{babel}

\addto\extrasenglish{%
}

\addto{\captionsenglish}{\renewcommand*{\appendixname}{Appendix: Omitted Proof Details}}

\usepackage[nosumlimits,nonamelimits]{amsmath}
\usepackage{amsthm,amssymb}

\usepackage[capitalize,noabbrev,nameinlink]{cleveref} %

\theoremstyle{plain}

\newtheorem{theorem}{Theorem}[section]

\newtheorem{proposition}[theorem]{Proposition}
\newtheorem{lemma}[theorem]{Lemma}
\newtheorem{corollary}[theorem]{Corollary}

\theoremstyle{definition}

\newtheorem{definition}[theorem]{Definition}
\newtheorem{example}[theorem]{Example}
\newtheorem{remark}[theorem]{Remark}

\usepackage{booktabs}   %
\usepackage{subcaption} %
\crefname{theorem}{Theorem}{Theorems}
\crefname{definition}{Definition}{Definitions}
\crefname{lemma}{Lemma}{Lemmas}
\crefname{corollary}{Corollary}{Corollaries}
\crefname{proposition}{Proposition}{Propositions}

\newcommand{\comma}{,\operatorname{}\linebreak[1]}

\DeclareMathOperator{\img}{img}

\usepackage[utf8]{inputenc}

\usepackage{hypcap}
\usepackage{microtype}

\usepackage[displaymath,mathlines,switch]{lineno}
 
\usepackage{etoolbox} %

\newcommand*\linenomathpatch[1]{%
	\cspreto{#1}{\linenomath}%
	\cspreto{#1*}{\linenomath}%
	\csappto{end#1}{\endlinenomath}%
	\csappto{end#1*}{\endlinenomath}%
}

\linenomathpatch{equation}
\linenomathpatch{gather}
\linenomathpatch{multline}
\linenomathpatch{align}
\linenomathpatch{alignat}
\linenomathpatch{flalign}

\ifdraft{
  \usepackage[draft]{showlabels}  \linenumbers
  \usepackage[layout=footnote,marginclue,draft]{fixme}
  \FXRegisterAuthor{sg}{asg}{SG}	%
  \FXRegisterAuthor{pn}{apn}{PN}	%
  \FXRegisterAuthor{ls}{als}{LS}	%
  \FXRegisterAuthor{pw}{apw}{PW}	%

}
{
  \newcommand{\sgnote}[1]{}
  \newcommand{\pnnote}[1]{}
  \newcommand{\lsnote}[1]{}
}

\usepackage{cancel}
\usepackage{bbm}
\usepackage{dsfont}
\usepackage{mathtools}
\usepackage[utf8]{inputenc}
\usepackage{tikz-cd}
\usetikzlibrary{decorations.markings}
\usepackage{pedromacros-clean}
\usepackage{centernot}

\tikzcdset{
  only tips/.style={/pgf/tips=true, /tikz/draw=none},
      Corner/.tip={Straight Barb[angle'=90, round, length=4.8pt,width=10pt]},
      Corner Tail/.tip={Corner[reversed, sep=+3pt +1.5]},
         rightcorner/.style={only tips, /tikz/arrows=-Corner},
          leftcorner/.style={only tips, /tikz/arrows=Corner Tail-},
     leftrightcorner/.style={only tips, /tikz/arrows=Corner-Corner},
}

\tikzstyle{shiftarr}=[
rounded corners,%
to path={--([#1]\tikztostart.center)
  -- ([#1]\tikztotarget.center) \tikztonodes
  -- (\tikztotarget)},
]

\tikzset{
  commutative diagrams/.cd,
  arrow style=tikz,
  diagrams={>={Straight Barb[length=1.75pt,width=3.85pt,inset=1.95pt]}}, %
  row sep=large,
  column sep = huge
}

\tikzset{cong/.style={draw=none,edge node={node [sloped, allow upside down, auto=false]{$\cong$}}},
  iso/.style={draw=none,every to/.append style={edge node={node [sloped, allow upside down, auto=false]{$\cong$}}}}}

\tikzcdset{scale cd/.style={every label/.append style={scale=#1},
    cells={nodes={scale=#1}}}}

\usetikzlibrary{
  fit,%
  calc,%
  arrows,%
  arrows.meta,%
  intersections,%
  shapes.misc,%
  shapes.arrows,%
  patterns,%
  automata,%
  chains,%
  matrix,%
  positioning,%
  scopes,%
  decorations.markings,%
  decorations.pathmorphing,%
  external,
  backgrounds
}

\usepackage{stackengine}
\stackMath
\newcommand\tsup[2][2]{%
  \def\useanchorwidth{T}%
  \ifnum#1>1%
    \stackon[-1.05ex]{\tsup[\numexpr#1-1\relax]{#2}}{\scalebox{2}[1]{$\mathchar"307E$}\kern-.5pt}%
  \else%
    \stackon[-.9ex]{#2}{\scalebox{2}[1]{$\mathchar"307E$}\kern-.5pt}%
  \fi%
}

\usepackage{caption}
\usepackage{subcaption}

\usepackage{xspace}

\usepackage{lineno}

\usepackage{etoolbox} %

\providecommand{\oname}[1]{{\operatorname{\mathsf{#1}}}}

\newcommand{\nat}{\mathbb{N}}

\newcommand{\dom}{\mathop{\oname{dom}}}
\newcommand{\dor}{\mathop{\oname{d}}}

\newcommand{\cod}{\mathop{\oname{img}}}
\newcommand{\cor}{\mathop{\oname{i}}}

\renewcommand{\img}{\mathop{\oname{img}}}

\newcommand{\subid}[1]{\lceil #1 \rceil}

\newcommand{\domr}[1]{\lceil\dom #1 \rceil} %
\newcommand{\codr}[1]{\lceil\cod #1 \rceil} %

\newcommand{\ftcF}{\functorfont{\underline{F}}}

\providecommand{\noqed}{\def\qed{}}			 %

\providecommand{\ito}{\rightarrowtail}                                       %

\providecommand{\xto}[1]{\,\xrightarrow{#1}\,}

\renewcommand{\xto}[1]{\mathrel{\raisebox{-1pt}{$\xrightarrow{\hspace{2.5pt}\smash{\raisebox{-1.5pt}{\makebox(3,0)[b]{\scriptsize~$#1$}}\hspace{2.5pt}}}$}}}

\providecommand{\xfrom}[1]{\,\xleftarrow{\;#1}\,}

\renewcommand{\xfrom}[1]{\mathrel{\raisebox{-1pt}{$\xleftarrow{\smash{\hspace{2.5pt}\raisebox{-1.5pt}{\makebox(3,0)[b]{\scriptsize~$#1$}}\hspace{2.5pt}}}$}}}
\providecommand{\mone}{{\text{\kern.5pt\rmfamily-}\mathsf{\kern-.5pt1}}}

\ifdraft{}
{
   \renewcommand{\todo}[1]{}
}

\usepackage{amsthm}

\makeatletter
\newcommand{\linebreakand}{%
  \end{@IEEEauthorhalign}
  \hfill\mbox{}\par
  \mbox{}\hfill\begin{@IEEEauthorhalign}
}
\makeatother

\begin{document}
\allowdisplaybreaks

\title{Relators and Notions of Simulation Revisited}

\author{
\IEEEauthorblockN{Sergey Goncharov}
\IEEEauthorblockA{
	\textit{University of Birmingham}\\
	Birmingham, UK \\
	s.goncharov@bham.ac.uk
}
\and
\IEEEauthorblockN{Dirk Hofmann}
\IEEEauthorblockA{
	\textit{CIDMA, University of Aveiro}\\
	\textit{Aveiro, Portugal} \\
	dirk@ua.pt
}
\and
\IEEEauthorblockN{Pedro Nora}
\IEEEauthorblockA{
	\textit{Radboud Universiteit}\\
    \textit{Nijmegen, Netherlands}\\
    pedro.nora@ru.nl
}

\linebreakand

\IEEEauthorblockN{Lutz Schr\"oder}
\IEEEauthorblockA{
	\textit{Friedrich-Alexander-Universität Erlangen-Nürnberg} \\
	\textit{Erlangen, Germany} \\
	lutz.schroeder@fau.de
}
\and
\IEEEauthorblockN{Paul Wild}
\IEEEauthorblockA{
	\textit{Friedrich-Alexander-Universität Erlangen-Nürnberg} \\
	\textit{Erlangen, Germany} \\
	paul.wild@fau.de
}
\thanks{
  Sergey Goncharov acknowledges support by the Deutsche Forschungsgemeinschaft (DFG) -- project number 501369690. %
  Lutz Schröder acknowledges support by the Deutsche Forschungsgemeinschaft (DFG) -- project number 531706730. %
  Paul Wild acknowledges support by the Deutsche Forschungsgemeinschaft (DFG) -- project number 434050016 %
}
}

\maketitle \thispagestyle{empty}

\begin{abstract}
  Simulations and bisimulations are ubiquitous in the study of concurrent systems and modal logics of various types. Besides classical relational transition systems, relevant system types include, for instance, probabilistic, weighted, neighbourhood-based, and game-based systems. \emph{Universal coalgebra} abstracts system types in this sense as set functors.  Notions of (bi)simulation then arise by extending the functor to act on relations in a suitable manner, turning it into what may be termed a \emph{relator}. We contribute to the study of relators in the broadest possible sense, in particular in relation to their induced notions of (bi)similarity. Specifically, (i) we show that every functor that preserves a very restricted type of pullbacks (termed 1/4-iso pullbacks) admits a sound and complete notion of bisimulation induced by the \emph{coBarr relator}; (ii) we establish equivalences between properties of relators and closure properties of the induced notion of (bi)simulation, showing in particular that the full set of expected closure properties requires the relator to be a lax extension, and that soundness of (bi)simulations requires preservation of diagonals; and (iii) we show that functors preserving inverse images admit a \emph{greatest} lax extension. In a concluding case study, we apply~(iii) to obtain a novel highly permissive notion of \emph{twisted bisimulation} on labelled transition systems.
\end{abstract}

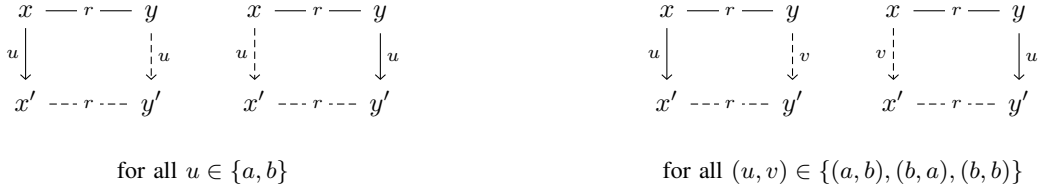
\begin{figure*}[!t]
\centering
\begin{subfigure}[b]{0.4\textwidth}
\centering
\begin{equation*}
\begin{tikzcd}[column sep=1.5em]
	x && y \\
	x' && y'
	\arrow["r"{description}, no head, from=1-1, to=1-3]
	\arrow["u"', from=1-1, to=2-1]
	\arrow["u", dashed, from=1-3, to=2-3]
	\arrow["r"{description}, dashed, no head, from=2-1, to=2-3]
\end{tikzcd}
\qquad
\begin{tikzcd}[column sep=1.5em]
	x && y \\
	x' && y'
	\arrow["r"{description}, no head, from=1-1, to=1-3]
	\arrow["u"', dashed, from=1-1, to=2-1]
	\arrow["u",  from=1-3, to=2-3]
	\arrow["r"{description}, dashed, no head, from=2-1, to=2-3]
\end{tikzcd}
\end{equation*}
     \caption*{\textnormal{for all~$u\in \{a,b\}$}}
 \end{subfigure}
 \hspace{2.75em}
 \begin{subfigure}[b]{0.4\textwidth}
     \centering
\begin{equation*}
\begin{tikzcd}[column sep=1.5em]
	x && y \\
	x' && y'
	\arrow["r"{description}, no head, from=1-1, to=1-3]
	\arrow["u"', from=1-1, to=2-1]
	\arrow["v", dashed, from=1-3, to=2-3]
	\arrow["r"{description}, dashed, no head, from=2-1, to=2-3]
\end{tikzcd}
\qquad
\begin{tikzcd}[column sep=1.5em]
	x && y \\
	x' && y'
	\arrow["r"{description}, no head, from=1-1, to=1-3]
	\arrow["v"', dashed, from=1-1, to=2-1]
	\arrow["u",  from=1-3, to=2-3]
	\arrow["r"{description}, dashed, no head, from=2-1, to=2-3]
\end{tikzcd}
\end{equation*}
     \caption*{for all~$(u,v)\in\{(a,b),(b,a),(b,b)\}$}
 \end{subfigure}
    \caption{Standard bisimulation (left) and alternative additional clauses for twisted bisimulation (right).}
    \label{fig:bis}
\end{figure*}
\section{Introduction}

State-based systems as used, for instance, in concurrency or in the
semantics of modal logics, come in many different flavours, and in
particular may vary in their branching type. While classically,
attention has focused on relational systems, in particular
nondeterministic ones such as labelled transition systems (LTS) or
Kripke frames, there has been long-standing interest in other types
where branching is probabilistic (e.g.~\cite{LarsenSkou91}), weighted
(e.g.~\cite{KupfermanSattlerEtAl02,BuchholzKemper10}),
game-based~\cite{AlurHenzingerEtAl02}, or neighbourhood-based as in the
Montague-Scott semantics of modal logic~\cite{Chellas80}, concurrent
dynamic logic~\cite{Peleg87}, or game logic~\cite{Parikh83}.

A unifying framework for such diverse system types is available in the
shape of \emph{universal coalgebra}~\cite{Rutten00}, in which the
system type is abstracted as a set functor, whose coalgebras then play
the role of systems. For instance, coalgebras for the powerset functor
are relational transition systems, and coalgebras for the distribution
functor are probabilistic transition systems. The coalgebraic
framework induces a canonical abstract notion of \emph{behavioural
  equivalence}, which instantiates to standard branching-time
equivalences (such as Park-Milner bisimilarity on LTS) in all
examples~\cite{Klin09}. Roughly speaking, two states are behaviourally
equivalent if they may be identified under suitable coalgebra
morphisms. The situation is more complicated regarding the task of
\emph{certifying} behavioural equivalence by means of (ideally, small)
witnessing relations, in generalization of bisimulations on Kripke
frames or LTS; the existence of such relations typically depends on
properties of the underlying functor. One of the first insights in
universal coalgebra was that \emph{Aczel-Mendler
  bisimulations}~\cite{AczelMendler89} are sound and complete for
behavioural equivalence if the functor preserves weak pullbacks; here,
we call a notion of bisimulation \emph{sound} if bisimilar states are
behaviourally equivalent, and \emph{complete} if the converse holds,
where as usual two states are bisimilar if they are related by some
bisimulation.

Generally, coalgebraic notions of bisimulation rely on extending the
action of the underlying functor to relations; we refer to such
extensions in the most general sense, essentially only requiring
well-typedness, as \emph{relators}, following Thijs~\cite{Thijs96}
(noting that the term has been used with different meanings in the
field, e.g.~\cite{BackhouseBruinEtAl91,Levy11}). For instance, the
above-mentioned notion of Aczel-Mendler bisimulation is induced by a
particular form of relator, the \emph{Barr
  extension}~\cite{Barr70}. Relators have been equipped with various
sets of axioms that ensure good properties of the induced class of
bisimulations. Notably, notions of bisimulation induced by
\emph{normal} (or \emph{diagonal-preserving}) \emph{lax
  extensions}~\cite{Thijs96,Levy11,MartiVenema12,MartiVenema15} are
sound and complete for behavioural equivalence and closed under
relational composition, and moreover allow for sound up-to techniques
such as bisimulation up to transitivity.

Depending on the exact set of required properties, suitable relations
and corresponding notions of bisimulation may or may not exist for
given functors and their associated system types. For instance, the
above-mentioned Barr extension is a lax extension iff the functor
weakly preserves pullbacks~\cite{Barr70,Trnkova80}. It has recently been
shown that normal lax extensions exist for functors that preserve
either inverse images or weakly preserve kernel pairs 
and so-called 1/4-iso pullbacks,
i.e.\ pullbacks in which at least one leg is isomorphic, and moreover
that preservation of 1/4-iso pullbacks is a necessary condition for
existence of a normal lax extension~\cite{GoncharovHofmannEtAl25}.

In the present work, we further analyse the landscape of relators and
lax extensions, in particular strengthening the perspective that a
given functor potentially comes with a whole range of relevant
relators and associated notions of \mbox{(bi)}simulation. In detail, our
contributions are as follows.
\begin{itemize}[wide]
\item We show that functors preserving 1/4-iso pullbacks admit a sound
  and complete notion of bisimulation induced by the \emph{coBarr
    relator} (\Cref{p:17e}), which however is not in general a lax extension.
\item We show that a given relator is a lax extension iff the induced
  class of (bi)simulations contains all coalgebra homomorphisms
  and their converses and is closed under composition
  (\cref{p:6}). Moreover, we show that normality is essentially
  a necessary condition for soundness of bisimulations
  (\cref{p:10d}).
\item We show that functors that preserve inverse images admit a
  greatest normal lax extension (\cref{thm:main}), and hence a
  maximally permissive notion of bisimulation.
\item In a concluding case study, we illustrate the latter result on
  functors of the form~$(-)^A$, which model deterministic~$A$-labelled
  transition systems. We arrive at a new notion of \emph{twisted
    bisimulation}, which we extend to obtain a notion of twisted
  bisimulation on LTS that is more permissive than the standard notion
  (\cref{sec:lts}). Both notions are illustrated for
 $A=\{a,b\}$ in \Cref{fig:bis}. The left hand diagram recalls the
  standard definition: A relation~$r$ between the state sets of two
  LTS is a bisimulation if whenever~$x\, r\, y$ for states~$x,y$, then
  the \emph{forth} condition (for every labelled transition
 $x\xto{u} x'$, there is~$y'$ such that~$x'\,r\,y'$ and
 ${y\xto{u} y'}$) and the symmetric \emph{back} condition hold. The
  definition of twisted bisimulation alternatively allows~$x,y$ to
  satisfy the right-hand clause in \Cref{fig:bis}, in which
  actions are mismatched. Explicitly, $r$ is a twisted bisimulation if
  whenever~$x\mathrel{r}y$, then either the standard forth and back
  conditions shown on the left hold, or the alternative clauses shown
  on the right hold for all~$(u,v)\in\{(a,b),(b,a),(b,b)\}$.  Indeed
  one can allow a third set of alternative clauses with the roles
  of~$a$ and~$b$ interchanged. Since twisted bisimulation is induced
  by a normal lax extension, it remains sound for standard
  bisimilarity while allowing for smaller bisimulation relations.

  \paragraph*{Related work} Up to variations in the axiomatics, lax
  extensions go back to work on coalgebraic
  simulation~\cite{HesselinkThijs00} (under the name \emph{relational
    extensions}). The axiom for lax preservation of composition first
  appears explicitly in work on simulation quotients~\cite{Levy11},
  where the term \emph{relator} is used. We have already mentioned
  work by Marti and Venema relating lax extensions to modal
  logic~\cite{MartiVenema12,MartiVenema15}; at the same time, Marti
  and Venema prove that the notion of bisimulation induced by a normal
  lax extension captures the standard notion of behavioural
  equivalence. \emph{Lax relation liftings}, constructed for functors
  carrying a coherent order structure~\cite{HughesJacobs04}, also
  serve the study of coalgebraic simulation but obey a different
  axiomatics than lax extensions (cf.~\cite[Remark
  4]{MartiVenema15}). There has been recent interest in quantitative
  notions of lax extensions that act on relations taking values in a
  quantale, such as the unit interval, in particular with a view to
  obtaining notions of quantitative
  bisimulation~\cite{HofmannSealEtAl14,Gavazzo18,WildSchroder20,WildSchroder22,GoncharovHofmannEtAl23}
  that witness low behavioural distance (the latter having first been
  treated in coalgebraic generality by Baldan et
  al.~\cite{BaldanBonchiEtAl18}). The correspondence between normal
  lax extensions and separating sets of modalities generalizes to the
  quantitative
  setting~\cite{WildSchroder20,WildSchroder22,GoncharovHofmannEtAl23}.
\end{itemize}
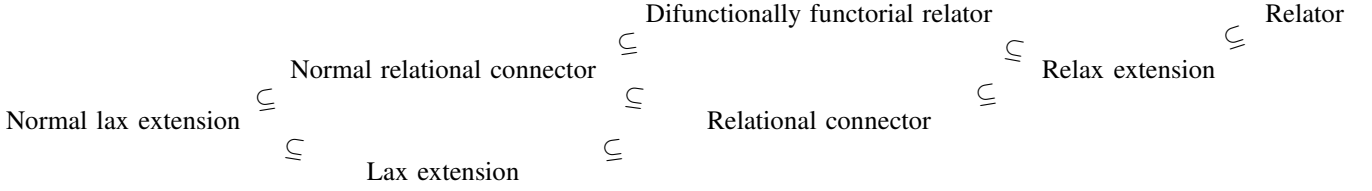
\begin{figure*}[!t]
\centering

\begin{equation*}
\begin{tikzcd}[column sep=1em, row sep=1ex]
&&\text{Difunctionally functorial relator}
	\ar[dr, phantom,"\qquad\subseteq"{sloped}]
&& \text{Relator}
\\
& \text{Normal relational connector}
	\ar[dr, phantom,"\subseteq"{sloped}]
	\ar[ur, phantom,"\subseteq"{sloped}]  
& & \text{Relax extension}
	\ar[ru, phantom,"~~\subseteq"{sloped}] 
\\
\text{Normal lax extension}
	\ar[ur, phantom,"\subseteq\qquad"{sloped}]
	\ar[dr, phantom,"\subseteq"{sloped}] & & \text{Relational connector}
	\ar[ur,phantom,"\subseteq"{sloped}]
\\
& 
\text{Lax extension}
	\ar[ur, phantom,"\subseteq"{sloped}]
\end{tikzcd}
\end{equation*}
\caption{Classes of Relators.}
\label{fig:taxonomy}
\end{figure*}

\section{Preliminaries: Relations and Coalgebras}
\label{sec:prelims}

We assume basic familiarity with category theory (e.g.~\cite{AdamekHerrlichEtAl90}).

Unless stated otherwise, we work in the category~$\SET$ of sets and functions throughout.
Another relevant category is the category~$\REL$ whose objects are again sets, but whose morphisms
are the corresponding relations.
We use the notation~$r \colon X \relto Y$ to indicate that~$r$ is a relation~$r\subseteq X\times Y$.
We say that~$r\colon X\relto Y$ and~${s \colon Y'\relto Z}$ are \df{composable} if~$Y=Y'$,
and extend this terminology to sequences of relations in the
obvious manner. Both for functions and for
relations, we use \emph{applicative composition}, i.e.\ given
$r\colon X\relto Y$ and~${s \colon Y\relto Z}$, their composite
$s\cdot r\colon X\relto Z$ is
$\{(x,z)\mid\exists y\in
Y.\,x\mathrel{r}y\mathrel{s}z\}$. Relations between the same sets are ordered by
inclusion, hence we write \(r \leq r'\) as a synonym to \(r \subseteq r'\). We denote by
$1_X\colon X\to X$ the identity map (hence relation) on~$X$, and we
say that a relation \(r \colon X \relto X\) is a \df{subidentity} if
\(r \leq 1_X\). The identity endofunctor will generally be denoted as~$\ftId$.

Given a relation~$r\colon X \relto Y$,
$r^\circ\colon Y \relto X$ denotes the corresponding converse
defined by \(y\mathrel{r}^{\circ}x\iff x\mathrel{r}y\); in particular, if~$f\colon X\to Y$ is a function, then
$f^\circ\colon Y\relto X$ denotes the converse of the corresponding
relation. For a relation~$r\colon X\relto Y$, we denote by
$\dom r\subseteq X$ and~$\cod r\subseteq Y$ its respective domain ($\dom r=\{x\in X\mid\exists y\in Y.\,x\mathrel{r}y\}$
and image ($\cod r=\{y\in Y\mid\exists x\in X.\,x\mathrel{r}y\}$). A special
class of relations of interest are \df{difunctional relations}
\cite{Riguet48}, which are relations factorizable as~$g^\circ\cdot f$ for
some functions~$f\colon X\to Z$ and~$g\colon Y\to Z$, i.e.\ \(x~r~y\)
iff \(f(x) = g(y)\).  The \df{difunctional closure} of a relation
$r\colon X\relto Y$ is the least difunctional relation~$\hat r \colon X \relto Y$ greater than or
equal to~$r$. More explicitly, $\hat r = \bigvee_{n\in\nat} r\cdot (r^\circ\cdot r)^n$ (e.g.~\cite{Riguet48,GummZarrad14}).

Given an endofunctor~$\ftF\colon\SET \to \SET$, an \df{$\ftF$-relator}, or simply a relator, is a monotone map~$\eR \colon \REL \to
\REL$ that sends a relation from~$X$ to~$Y$ to a relation from~$\ftF X$ to~$\ftF
Y$.  (See~\cite{Thijs96,BackhouseBruinEtAl91,Levy11} for other uses of the term ``relator''). 
We say that a relator~$\eR \colon \REL \to \REL$ is \df{symmetric} if
$\eR(r^\circ) = (\eR r)^\circ$ and \df{normal} if~$\eR 1_X = 1_{\ftF X}$.
We order~$\ftF$-relators pointwise, i.e, given~$\ftF$-relators~$\eR_1$ and~$R_2$, we write~$\eR_1 \leq R_2$ if for every relation~$r$, $\eR_ 1 r \leq \eR_2 r$.

A \df{relax extension} of~$\ftF$~\cite{GoncharovHofmannEtAl23} is a relator that satisfies
$\ftF f \le \eR f$ and~${(\ftF f)}^\circ\le \eR (f^\circ)$ for all functions~$f$.
A \df{lax extension} is such a relax extension~$\eR$ that satisfies \emph{lax preservation of composition}:
$\eR s\cdot R r\le \eR (s\cdot r)$ for all \(r \colon X \relto Y\), \(s \colon Y \relto Z \).
A relator~$\eR \colon \REL \to \REL$ is a \df{relational connector}~\cite{RelationalConnectors} if for every set~$X$, $1_{\ftF X} \leq \eR 1_X$ and for every relation~$r \colon X \relto Y$, and all function~$f \colon A \to X$ and~$g \colon B \to Y$, $\eR( g^\circ \cdot r \cdot f) = {(\ftF g)^\circ \cdot \eR r \cdot \ftF f}$, a condition called \df{naturality}.
It is well-known that every lax extension is a relational connector (e.g.~\cite{HofmannSealEtAl14}), and it is easy to see that every relational connector is a relax extension. 
The resulting situation is summarized in~\autoref{fig:taxonomy} (inclusion signs indicate relations between the corresponding classes). 
This also involves a technical notion of \df{difunctionally functorial relator}, which are such relators~$\eR$ 
that $\eR(g^\circ \cdot f) = (\ftF g)^\circ \cdot \ftF f$, for all functions~$f \colon X \to A$ and~$g \colon Y \to A$ -- the corresponding inclusions 
in~\autoref{fig:taxonomy} follow.

Given a functor \(\ftF \colon \SET \to \SET\), we will be interested in its (weak)
pullback preservation properties, specifically in (weak) preservation of pullbacks of the
form:
\begin{center}
	\begin{tikzcd} %
		P        \ar[from=1-1, to=2-2, phantom, very near start, "\lrcorner"] & B \\
		X & Y
		\ar[from=1-1, to=1-2]
		\ar[from=1-1, to=2-1, tail, "f"']
		\ar[from=1-2, to=2-2, "g"]
		\ar[from=2-1, to=2-2]
	\end{tikzcd}
\end{center}
where~$f$ is a mono (which are pullbacks iff they are weak pullbacks, so weak preservation and preservation of such pullbacks coincide).
We say that~$\ftF$ \df{preserves 1/4-iso pullbacks} if~$\ftF$ preserves the above pullbacks when~$f$ is an isomorphism,  and that~$\ftF$ \df{preserves inverse images} if~$\ftF$ preserves the above pullbacks when~$g$ is a monomorphism. Both properties 
are significantly weaker than weak pullback preservation, i.e.\ preservation of 
weak pullbacks by~$\ftF$~\cite{Rutten00}.

An \df{$\ftF$-coalgebra}~$(X,\alpha)$ for an endofunctor~$\ftF\colon\SET\to\SET$
consists of a set~$X$ of \df{states} and a \df{transition map}~${\alpha\colon X\to\ftF X}$.
A \df{morphism}~$f\colon (X,\alpha)\to(Y,\beta)$ of~$\ftF$-coalgebras is a map
$f\colon X\to Y$ for which~$\beta\cdot f={\ftF f\cdot\alpha}$. Such morphisms are
thought of as preserving the behaviour of states, and correspondingly,
states~$x$ and~$y$ in coalgebras~$(X,\alpha)$ and~$(Y,\beta)$, respectively,
are \df{behaviourally equivalent} if there exist a coalgebra~$(Z,\gamma)$ and
morphisms~$f\colon(X,\alpha)\to(Z,\gamma)$, $g\colon(Y,\beta)\to(Z,\gamma)$
such that~$f(x)=g(y)$.

\section{Simulations induced by Relators}
\label{sec:sim-ind-rel}

We fix a functor~$\ftF \colon \SET \to \SET$ for the remainder of the
technical development until the end of \Cref{sec:greatest-lax-extension}.

Relators induce notions of simulation and similarity between (possibly distinct)
$\ftF$-coalgebras.
Given a relator~$\eR \colon \REL \to \REL$, a relation~$r \colon X \relto
	Y$ is an \df{$\eR$-simulation} from a coalgebra~$\alpha \colon X \to \ftF X$ to a
coalgebra~$\beta \colon Y \to FY$~if\sgnote{Maybe, draw it as a lax diagram?}
\begin{displaymath}
	r \leq \beta^\circ \cdot \eR r \cdot \alpha,
\end{displaymath}
i.e, if~$x\mathrel{r}y$ entails~$\alpha(x)\, \eR r\, \beta(y)$, for all~$x
	\in X$ and~$y \in Y$. When we speak of \df{notions of simulation}, we always understand them as being induced by a relator in this way.
Relators are naturally composed.
In particular, given an~$\ftF$-relator~$\eR$ and a relator~$\ftU$ of the identity functor on~$\SET$, $\eR \cdot \ftU$ yields an~$\ftF$-relator, and we say that an~$\eR \cdot \ftU$-simulation is an~$\eR$-simulation \df{up-to}~$\ftU$.
In particular, whenever $\ftU$ is the relator of the identity functor on~$\SET$ that sends a relation to its difunctional closure, we say that an $\eR \cdot \ftU$-simulation is an $\eR$-simulation up-to difunctionality. 
\pnnote{Up-to or up to?}
Whenever~$\eR_1$ and~$\eR_2$ are~$\ftF$-relators such that~$\eR_1 \leq \eR_2$, then every~$\eR_1$-simulation is an~$\eR_2$-simulation; in the running discussion, we phrase this by saying that~$\eR_2$ induces a \df{more permissive} notion of simulation than~$\eR_1$.
As we assume relators to be monotone, given coalgebras~$\alpha \colon X \to \ftF X$ and~$\beta \colon Y \to \ftF Y$, the map sending a relation~$r \colon X \relto Y$ to the relation~$\beta^\circ \cdot \eR r \cdot \alpha\colon X \relto Y$ is monotone. Therefore, by the Knaster-Tarski fixed point theorem, this map has a greatest fixed point, the greatest~$\eR$-simulation from~$\alpha$ to~$\beta$, which we call \df{$\eR$-similarity} from~$\alpha$ to~$\beta$. 
When~$\eR$ is symmetric, we often speak about~\df{$\eR$-bisimulations} instead of~$\eR$-simulations, and about~\df{$\eR$-bisimilarity} instead of~$\eR$-similarity;
this will be justified in \Cref{p:6}.
For every relator~$\eR \colon \REL \to \REL$, if for all~$\ftF$-coalgebras~$\alpha$ and~$\beta$, $\eR$-similarity from~$\alpha$ to~$\beta$ is greater than or
equal to behavioural equivalence from~$\alpha$ to~$\beta$, then we say that
$\eR$-similarity is \df{complete}. Conversely, if~$\eR$-similarity from~$\alpha$ to~$\beta$ is
less than or equal to behavioural equivalence from~$\alpha$ to~$\beta$, then we
say that~$\eR$-similarity is \df{sound}.

In this section, we study general properties of notions of simulation.
We begin by providing sufficient conditions for soundness and for completeness.
The next result  substantially generalizes~\cite[Theorem 1]{MartiVenema15}, which states that every symmetric normal lax extension induces a sound and complete notion of bisimilarity.
Throughout the paper we will see examples of relators, such as the coBarr relator, that satisfy the conditions of the theorem below but are not symmetric lax extensions.

\begin{theorem}
	\label{p:4}
	Let~$\eR$ be an~$\ftF$-relator.
	\begin{enumerate}[wide]
		\item If for all functions~$f \colon X \to A$ and~$g \colon Y \to A$, ${\eR(g^\circ
			      \cdot f)} \geq (\ftF g)^\circ \cdot \ftF f$, then~$\eR$-similarity is complete.
		\item If~$\eR$ preserves difunctional relations and for every epi-cospan~$(f \colon X \to A, g \colon Y \to A) \in \SET$, $\eR(g^\circ \cdot f) \leq (\ftF g)^\circ \cdot \ftF f$, then
		     $\eR$-similarity is sound.
	\end{enumerate}
\end{theorem}
\begin{proof}
	\begin{enumerate}[wide]
		\item Let~$(X,\alpha)$ and~$(Y,\beta)$ be~$\ftF$-coalgebras, and let~$x \in X$ and~$y
			      \in Y$ be behaviourally equivalent elements. Then, there are coalgebra
		      homomorphisms~$f \colon (X,\alpha) \to (Z, \gamma)$ and~$g \colon (Y,\beta) \to
			      (Z,\gamma)$ such that~$f(x) = g(y)$. Hence, with~$r = g^\circ \cdot f$, we have
		     $x\, r \, y$ and as~$f,g$ are coalgebra homomorphisms, by hypothesis, we obtain
		     $r \leq \beta^\circ \cdot (\ftF g)^\circ \cdot \ftF f \cdot \alpha \leq \beta^\circ
			      \cdot \eR r \cdot \alpha$. Therefore, $x$ and~$y$ are~$\eR$-similar.
		\item Let~$r \colon X \relto Y$ be an~$\eR$-simulation from~$(X,\alpha)$ to~$(Y,\beta)$
		      and~$(\pi_1 \colon A \to X,\pi_2 \colon A \to Y)$ be a span in~$\SET$ such that
		     $r=\pi_2 \cdot \pi_1^\circ$. Consider the pushout~$(p_1 \colon X \to O\comma p_2
			      \colon Y \to O)$ of~$(\pi_1,\pi_2)$. Then, $p_2^\circ \cdot p_1$ is
		      the difunctional closure~$\hat{r}$ of~$r$ and, in particular, for every~$x \in
			      X$ and~$y \in Y$, if~$x\mathrel{r}y$, then~$p_1(x) = p_2(y)$. Hence, the assertion
		      follows once we show that there is a coalgebra~$(O,\gamma)$ s.t.~$p_1 \colon
			      (X,\alpha) \to (O,\gamma)$ and~$p_2 \colon (Y,\beta) \to (O,\gamma)$ are
		      coalgebra homomorphisms. To this end, put~$\gamma(o) = \ftF p_1 \cdot
			      \alpha(x)$, if~$p_1(x)=o$, for some~$x \in X$ and~$\gamma(o) = \ftF p_2 \cdot
			      \beta(y)$, if~$p_2(y)=o$, for some~$y \in Y$. Note that~$(p_1,p_2)$ is an
		      epi-cocone, so we are assigning an element of~$\ftF O$ to every element of~$O$.
		      We claim that this assignment is well-defined. To see this, we begin by
		      observing that as~$\eR$ is monotone, $r \leq \beta^\circ \cdot \eR r \cdot \alpha
			      \leq \beta^\circ \cdot \eR \hat{r} \cdot \alpha$, and, as~$\eR$ preserves
		      difunctional relations, $\beta^\circ \cdot \eR \hat{r} \cdot \alpha$ is a
		      difunctional relation greater than or equal to~$r$, which entails, by definition
		      of difunctional closure, $\hat{r} \leq \beta^\circ \cdot \eR\hat{r} \cdot
			      \alpha$; i.e, $\hat{r}$ is an~$\eR$-simulation. Hence, by hypothesis, we obtain
		     $\hat{r} \leq \beta^\circ \cdot (\ftF p_2)^\circ \cdot \ftF p_1 \cdot \alpha$,
		      i.e, for all~$x \in X$ and~$y \in Y$, if~$p_1(x) = p_2(y)$, then~$\ftF p_1
			      \cdot \alpha (x) = \ftF p_2 \cdot \beta(y)$. Furthermore, by the way that
		      pushouts are constructed in~$\SET$, we have:
		      \pnnote{@Pedro: probably better to add more details here.}
		      \begin{enumerate}
			      \item for all~$x,x' \in X$ such that~$x \neq x'$, $p_1(x) = p_1(x')$ iff there is~$y \in
				            Y$ such that~$p_1(x) = p_2(y) = p_1(x')$;
			      \item for all~$y,y' \in Y$ such that~$y \neq y'$ , $p_2(y) = p_2(y')$ iff there is~$x \in
				            X$ such that~$p_2(y) = p_1(x) = p_2(y')$.
		      \end{enumerate}
		      Now, the claim follows straightforwardly by case distinction. Moreover, by
		      definition of~$\gamma$ we obtain immediately that~$p_1 \colon (X,\alpha) \to
			      (O,\gamma)$ and~$p_2 \colon (Y,\beta) \to (O,\gamma)$ are coalgebra
		      homomorphisms. Therefore, $\eR$-similarity is sound.\qed
	\end{enumerate}\noqed
\end{proof}
Recall that an~$\ftF$-relator $\eR$ is difunctionally functorial if for all functions~$f \colon X \to A$ and~$g \colon Y \to A$, $\eR(g^\circ \cdot f) = (\ftF g)^\circ \cdot \ftF f$.
\begin{corollary}
	\label{p:1}
	Let~$\eR$ be an difunctionally functorial $\ftF$-relator. Then, $\eR$-similarity is sound and complete.
\end{corollary}

The notion of relator is quite liberal and, naturally, the properties of the
corresponding class of simulations can vary significantly from one relator
to the other. In this work we think of simulations as witnesses for
behavioural equivalence and, therefore, as coalgebra homomorphisms are thought
to be ``behaviour preserving maps'', we are interested in  relators whose
classes of simulations contain all coalgebra homomorphism and their
converses since behavioural equivalence is a symmetric relation.
And, among such relators, we are particularly interested in those
whose corresponding classes of simulations are closed under composition or, at least, under composition with coalgebra homomorphisms.
Since empty relations are trivially simulations, the relators whose class of simulations satisfies certain properties are often characterized by their action on non-empty relations.
This phenomenon can be observed in the next results.

\begin{theorem}
	\label{p:6}
	Let~$\eR$ be an~$\ftF$-relator.
	The class of all~$\eR$-simulations
	\begin{enumerate}
		\item contains all homomorphisms of\/~$\ftF$-coalgebras and their converses iff for
		      every non-empty function~$f$, $\ftF f \leq \eR f$ and~$(\ftF f)^\circ \leq \eR (f^\circ)$;
		\item \label{p:3} is closed under composition iff for all relations~$r \colon X \relto Y$ and~$s \colon Y \relto Z$ such that~$s \cdot r$ is non-empty, $\eR s \cdot \eR r \leq \eR(s \cdot r)$;
		\item is closed under converses iff for every non-empty relation~$r \colon X \relto
			      Y$, $\eR(r^\circ) = (\eR r)^\circ$.
	\end{enumerate}
\end{theorem}
\begin{proof}
	\pnnote{@Pedro: add remaining proofs if time permits}
	We show~(\ref{p:3}), the other claims follow analogously.
  The fact that the condition is sufficient follows from standard arguments (e.g. \cite[Lemma 2]{Levy11}).
	Suppose that~$r \colon X \relto Y$ is an~$\eR$-simulation from~$(X,\alpha)$ to~$(Y,\beta)$ and~$s \colon Y \relto Z$ is an~$\eR$-simulation from~$(Y,\beta)$ to~$(Z,\gamma)$.
	If~$s \cdot r$ is empty, then it is trivially an~$R$-simulation.
	On the other hand, suppose that~$s \cdot r$ is non-empty.
	Then,  as~$r$ and~$s$ are~$\eR$-simulations, by hypothesis, $s \cdot r \leq \gamma^\circ \cdot \eR s \cdot \beta \cdot \beta^\circ \cdot \eR r \cdot \alpha \leq \gamma^\circ \cdot \eR(s \cdot r) \cdot \alpha$.
	Therefore, $s \cdot r$ is an~$\eR$-simulation.
	To see that the converse statement holds,  let~$r \colon X \relto Y$ and~$s \colon Y \relto Z$ be relations such that~$s \cdot r$ is non-empty.
	Now let~$u \in \ftF X$, $v \in \ftF Y$,~$w \in \ftF Z$ such that~$u\, \eR r\, v$ and~$v\, \eR s\, w$.
	Consider the coalgebras~$u \colon X \to \ftF X$, $v \colon Y \to \ftF Y$ and~$w \colon Z \to \ftF Z$ given by constant maps into~$u,v,w$, respectively.
	Then, $r$ is an~$\eR$-simulation from~$u$ to~$v$ and~$s$ is an~$\eR$-simulation from~$v$ to~$s$.
	Hence, by hypothesis, $s \cdot r$ is an~$\eR$-simulation, which means that~$s \cdot r \leq w ^\circ \cdot \eR(s \cdot r) \cdot u$. Thus, since $s \cdot r$ is non-empty, 
$w ^\circ \cdot \eR(s \cdot r) \cdot u$ is also non-empty, so $u\, \eR(s\cdot r)\, w$.
\end{proof}

A key feature of behavioural equivalence is that it is invariant under coalgebra homomorphisms.
For~$\eR$-similarity, this typically follows by showing that the class of~$\eR$-simulations is closed under (pre/post)composition with coalgebra homomorphisms.

\begin{theorem}
	Let~$\eR$ be an~$\ftF$-relator.
	The class of all~$\eR$-simulations
	\begin{enumerate}
		\item is closed under postcomposition with coalgebra homomorphisms iff for every
		      relation~$r \colon X \relto Y$ and every function~$f \colon A \to X$ such that~$r
			      \cdot f$ is non-empty, $\eR r \cdot \ftF f \leq \eR(r\cdot f)$;
		\item \label{p:8} is closed under postcomposition with converses of coalgebra homomorphisms iff for every
		      relation~$r \colon X \relto Y$ and every function~$f \colon X \to A$ such that~$r \cdot f^\circ$ is non-empty,
			$\eR r \cdot (\ftF f)^\circ \leq \eR (r\cdot f^\circ)$;
		\item is closed under precomposition with coalgebra homomorphisms iff for every
		      relation~$r \colon X \relto Y$ and every function~$g \colon Y \to B$ such that
		     $g \cdot r$ is non-empty, $\ftF g \cdot \eR r \leq \eR (g \cdot r)$;
		\item is closed under precomposition with converses of coalgebra homomorphisms iff for every
		      relation~$r \colon X \relto Y$ and every function~$g \colon B \to Y$ such that
		     $g^\circ \cdot r$ is non-empty, $(\ftF g)^\circ \cdot \eR r \leq \eR(g^\circ
			      \cdot r)$.
	\end{enumerate}
\end{theorem}
\begin{proof}
	Analogous to the proof of \Cref{p:6}.
	\pnnote{@Pedro: complete if time permits}
\end{proof}

The previous results motivate the study of simulations induced by (re)lax
extensions and relational connectors (cf.~\cref{sec:prelims}).
The prototypical example of a relax extension is the Barr
relator~\cite{Barr70}, which is normal and symmetric, and is a lax
extension iff it is a relational connector iff the functor preserves weak pullbacks.
The \df{Barr} relator~$\ftbF \colon \REL \to \REL$ of a functor~$\ftF \colon \SET \to \SET$ is defined
as follows. Given a relation~$r \colon X \relto Y$, take a span~${(\pi_1 \colon
	A \to X, \pi_2 \colon A \to Y)}$ such that~$r = \pi_2 \cdot
	\pi_1^\circ$. Then, put~$\ftbF r = \ftF \pi_2 \cdot (\ftF \pi_1)^\circ$.

\begin{example}
	\label{p:18d}
    The powerset functor~$\ftP \colon \SET \to \SET$ weakly preserves pullbacks. Its Barr extension coincides with the well-known Egli-Milner extension: Given~$r\colon X\relto Y$,
   $S\in\ftP X$, and~$T\in\ftP Y$, we have~$S\mathrel{Lr}T$ iff for
    all~$x\in S$, there exists~$y\in T$ such that~$x\mathrel{r}y$, and
    symmetrically. 
    For relational transition systems, understood as~$\ftP$-coalgebras, a~$\overline{\ftP}$-bisimulation is then just a bisimulation in the standard sense.
\end{example}

The Barr relator is the standard relator used to define a notion of bisimulation coalgebraically.
Regarding soundness and completeness, it is known that, independently of the functor, Barr-bisimilarity is sound~\cite{AczelMendler89}, and that for functors that weakly preserve pullbacks, Barr-bisimilarity is complete~\cite{MartiVenema15}.
However, \cref{p:1} suggests a different canonical construction to obtain sound and complete notions of similarity that is dual to the Barr relator on difunctional relations. Given a difunctional relation~$r
	\colon X \relto Y$, take a cospan~$(p_1 \colon X \to O\comma p_2 \colon Y \to O) \in
	\SET$ such that~$r = p_2^\circ \cdot p_1$. Then, put~$\ftcF r = (\ftF
	p_2)^\circ \cdot \ftF p_1$. Of course, for such construction to be
well-defined, $\ftcF r$ must be independent of the choice of the cospan,
and it has been shown recently~\cite{GoncharovHofmannEtAl25} that this is equivalent to the functor
preserving 1/4-iso pullbacks, which is equivalent to the functor being \df{monotone on
	difunctional relations} in the following sense: For all difunctional relations
factorized as \(g^\circ \cdot f \colon X \relto Y\) and \(g'^\circ \cdot f'
\colon X \relto Y\), if \(g^\circ \cdot f \leq g'^\circ \cdot f'\) then \((\ftF
g)^\circ \cdot \ftF f \leq (\ftF g')^\circ \cdot \ftF f'\). This means that whenever
~$\ftF$ preserves 1/4-iso pullbacks, sending
a relation~$r \colon X \relto Y$ to~$\ftcF \hat{r}$ as described above, with~$\hat{r}$ being the difunctional closure of~$r$,
defines an~$\ftF$-relator, which we call the \df{coBarr relator} of~$\ftF$ and denote by $\ftcF$. In the
sequel we record some properties of this construction that follow straightforwardly from the definition.

\begin{proposition}
	\label{p:5}
	Suppose that~$\ftF$ preserves 1/4-iso pullbacks.
	Then:
	\begin{enumerate}
		\item \label{p:15e} $\ftcF$ is a symmetric difunctionally functorial relator.
		\item \label{p:8e} If~$\ftF$ weakly preserves pullbacks, then for every relation~$r$, $\ftcF r = \ftbF \hat{r}$, where~$\hat{r}$ denotes the difunctional closure of~$r$.
		\item \label{p:10} $\ftbF \leq \ftcF$, and for every normal relational connector~$\eR$ of\/~$\ftF$,$\ftbF \leq \eR \leq \ftcF$, .
	\end{enumerate}
\end{proposition}

Regarding coBarr-bisimulations, from \Cref{p:1} and \Cref{p:5}(\ref{p:8e}) we obtain:

\begin{corollary}
	\label{p:17e}
	Suppose that~$\ftF$ preserves 1/4-iso pullbacks.
	Then:
		\begin{enumerate}
			\item \label{p:9} $\ftcF$-bisimilarity is sound and complete.
			\item If~$\ftF$ weakly preserves pullbacks, then the~$\ftcF$-bisimulations are precisely the~$\ftbF$-bisimulations up to difunctionality. 
		\end{enumerate}
\end{corollary}

\begin{remark}
	As far as we know, the coBarr relator was first proposed by Kurz in a private communication to Hansen et al.~\cite{HansenKupkeEtAl07}.
	We note that Hansen et al.\ do not require the functor to preserve 1/4-iso pullbacks, instead, given a relation~$r \colon X \relto Y$, they define~$\ftcF r$ as we have done here but w.r.t a pushout of the canonical span that determines~$r$.
	This construction is independent of the pushout and defines a relator for every functor, and, in particular, they show that coBarr-bisimilarity is sound.
	\cref{p:17e}(\ref{p:9}) complements their result by showing that coBarr-bisimilarity is complete under the assumption that the functor preserves 1/4-iso pullbacks.
\end{remark}

Preserving 1/4-iso pullbacks is a significantly weaker condition than weakly preserving pullbacks, which means that coBarr-bisimulations provide a sound and complete proof method for behavioural equivalence for a much larger class of functors than what is currently known generically for Barr-bisimulations.
However, a major drawback of the coBarr relator is that it is rarely a lax extension or even a relational connector, and, therefore, as we have seen in~\Cref{p:6}, the class of coBarr-bisimulations is rarely closed under composition.

\begin{example}
	\label{p:666}
	Even for the coBarr relator~$\functorfont{\underline{Id}}$ of the identity functor~$\ftId$, lax preservation of relational composition fails.
	(Of course, this functor by itself has trivial behavioural equivalence, but it serves as a building block in composite functors where triviality disappears, e.g.\ $2\times\ftId$.) 
	  Consider the endorelation~$z=\{(a,a), (b,a),(b,b)\}$ on the set~$2=\{a,b\}$, the map~${a \colon 1 \to 2}$ that selects the element~$a$, and the map~$!_2 \colon 2 \to 1$.
	The coBarr relator~$\functorfont{\underline{Id}}$ sends every relation to its difunctional closure.
	Therefore, $\functorfont{\underline{Id}} (z \cdot a) = a$ but~$\functorfont{\underline{Id}} z \cdot \functorfont{\underline{Id}} a = !_2^\circ$.
	The fact that the class of~$\underline{\ftId}$-bisimulations fails to be closed under composition often implies the same for the class of~$\eR$-bisimulations up-to difunctionality. 
	For instance, for the functor~$2 \times \ftId$, the class of Barr-bisimulations up-to difunctionality is not closed under composition. 
\end{example}

Nevertheless, an important consequence of \cref{p:5} is that the class of simulations induced by a normal relational connector is contained in the class of coBarr-simulations.
Hence, for a functor that preserves 1/4-iso pullbacks, the notion of coBarr-simulation is sound and complete, and  it is more permissive than any notion of simulation induced by a normal relational connector.
It has been shown that preserving 1/4-iso pullbacks is a necessary condition for a functor to admit a normal relational connector~\cite{GoncharovHofmannEtAl25}\, and, therefore, coBarr-simulations help us understand the simulations induced by normal relational connectors, such as the simulations induced by normal lax extensions, which are closed under composition.
Our interest in simulations induced by normal relational connectors stems from the following application of \cref{p:1}.

\begin{corollary}
	\label{p:6e}
	Every notion of similarity induced by a relational connector of a set functor is complete, and it is sound if the relational connector is normal.
\end{corollary}

In particular, this result shows that symmetry of lax extensions is not necessary to obtain sound and complete notions of similarity induced by lax extensions as assumed by Marti and Venema~\cite{MartiVenema15};
we will see how to construct non-symmetric normal lax extensions for exponential functors in \cref{sec:lts}, and
a concrete example is given in \Cref{p:600}(\ref{p:20e}).
This fact may be counter-intuitive at first sight, but the reason is that  on difunctional relations, normal lax extensions coincide with the coBarr relator, which is symmetric.
And, as we have seen in~\Cref{p:4}, to show soundness and completeness we only need to inspect the action of a relator on difunctional relations.

This motivates us to understand next when normality of a relational connector is necessary for soundness:

\begin{theorem}
	\label{p:2f}
	Suppose that~$\ftF$ admits a terminal coalgebra $(Z,\gamma)$.
	Let~$\eR$ be a relational connector of\/~$\ftF$.
		\begin{enumerate}
			\item $\eR$-similarity is sound iff\/~$\eR 1_Z = 1_{\ftF Z}$.
			\item If\/~$\ftF$ preserves monomorphisms, then~$\eR$-similarity is sound iff\/ $\eR 1_A = 1_{\ftF A}$ for every set~$A$ of cardinality at most~$|Z|$.
		\end{enumerate}
\end{theorem}
\noindent We stress that preserving monomorphisms is a very mild condition.
In fact, every set functor is ``naturally isomorphic up to $\varnothing$'' to a functor that preserves monomorphisms \cite{Barr93} and, hence, the category of coalgebras of every set functor is isomorphic to the category of coalgebras of a set functor that preserves monomorphisms.
\begin{proof}
	\begin{enumerate}[wide]
		\item Suppose that $\eR$-similarity is sound. 
			It is well-known that two states in a terminal coalgebra are behaviorally equivalent iff they are equal. 
			Thus, $\eR$-similarity from $(Z,\gamma)$ to itself is equal to $1_Z$ since $1_Z$ is simulation from $(Z,\gamma)$ to itself, as $1_{\ftF Z} \leq \eR 1_Z$ by definition of relational connector.
			Hence, as $\eR$-similarity from $(Z,\gamma)$ to itself is a fixed point of the map $\gamma^\circ \cdot - \cdot \gamma \colon \REL(Z,Z) \to \REL(Z,Z)$ that sends a relation $r$ to $\gamma^\circ \cdot r \cdot \gamma$, we obtain $1_Z = \gamma^\circ \cdot \eR 1_Z \cdot \gamma$.
			Therefore, as $\gamma$ is an isomorphism, $\eR 1_Z = 1_{\ftF Z}$.
			\lsnote{I confess I do not immediately see the argument}
			On the other hand, the condition is sufficient since behavioural equivalence is invariant under coalgebra homomorphisms and so is $\eR$-similarity for a relational connector~\cite[Lemma 4.2]{RelationalConnectors}.
		\item Let~$(X,\alpha)$ be an~$\ftF$-coalgebra where $|X|\le |Z|$.
				Then, there is an injective map~$i \colon X \rightarrowtail Z$.
				Hence, we have~$1_X = i^\circ \cdot 1_Z \cdot i$, and, as~$\eR$ is a relational connector, we obtain~$\eR 1_X = (\ftF i) ^\circ \cdot \eR 1_Z \cdot \ftF i$.
				Therefore, by the previous item, $\eR 1_X = (\ftF i) ^\circ \cdot \ftF i$, and since~$\ftF$ preserves monomorphisms, $\eR 1_X = 1_{\ftF X}$.\qed
	\end{enumerate}\noqed
\end{proof}

\Cref{p:2f} shows that normality is in general needed only on sufficiently small sets.
In fact, the next example shows that there are functors that do not admit a normal relational connector but even admit a lax extension that induces a sound and complete notion of similarity.

\begin{example}
	Consider the functor \((-)^2/\Delta\) that sends a set \(X\) to the quotient of the set~$X^2$ by the equivalence relation that identifies exactly the elements of the diagonal of~$X \times X$, and acts in the obvious way on functions.
	This functor does not preserve 1/4-iso pullbacks and, hence, does not admit a normal relational connector~\cite{GoncharovHofmannEtAl25}.
	However, since it preserves the terminal object, behavioural equivalence is trivial, therefore, its greatest lax extension which sends a relation~$ r \colon X \relto Y$ to the relation~$\ftF X \times \ftF Y$ induces a sound and complete notion of similarity.
\end{example}

On the other hand, the next result shows that normality is typically
necessary.

\begin{definition}
  A functor~$\ftG \colon \SET \to \SET$ is \df{$\zeta$-bounded} if it admits a terminal coalgebra~$Z \to \ftG Z$ and for every set~$X$ and every pair of elements~$u,v \in \ftG X$ there are a set~$A$ of cardinality $|A|\le|Z|$ and an injective map~$i \colon A \rightarrowtail X$ such that~$u,v \in \img(\ftG i)$.
\end{definition}
\noindent (Except in the corner case that~$Z$ is finite, one can
rephrase this definition as saying that~$\ftG$ has a final
coalgebra~$Z$ and is~$\kappa$-accessible where~$\kappa$ is the
cardinal successor of~$|Z|$.)

\begin{example}
  \label{p:14e}
  The following set functors are~$\zeta$-bounded.
  \begin{enumerate}
  \item Every constant functor.
  \item Every finitary functor with an infinite terminal coalgebra, e.g., the finite powerset functor and the finite multiset functor.
  \end{enumerate}
  \label{p:13e}
  Failures of~$\zeta$-boundedness typically relate to triviality of
  behavioural equivalence: If a set functor satisfies~$\ftF 1=1$,
  then~$1$ is the terminal coalgebra, so all states are behaviourally
  equivalent. For instance, the identity functor, more generally all
  exponential functors~$(-)^A$ except the constant functor
 $(-)^\emptyset$, and the discrete distribution functor all fail to
  be~$\zeta$-bounded for this reason.~$\zeta$-Boundedness is typically
  reinstated when such functors are combined with others to allow for
  actual observations; e.g.\ $2\times\ftId$ is~$\zeta$-bounded. More
  generally, we show in the appendix that all polynomial functors
  except the exponential functors~$(-)^A$ are~$\zeta$-bounded.\lsnote{Pull definition forward}
\end{example}

\begin{theorem}
	\label{p:10d}
	Suppose that $\ftF$ is $\zeta$-bounded and preserves monomorphisms.
	Let~$\eR$ be a relational connector of $\ftF$.
	If~$\eR$-similarity is sound, then~$\eR$ is normal.
\end{theorem}
\begin{proof}
	Let~$(Z,\gamma)$ be a terminal~$\ftF$-coalgebra, and let~$(X,\alpha)$ be an~$\ftF$-coalgebra.
	If the cardinality of~$X$ is less or equal than the cardinality of~$Z$, then, by \Cref{p:2f}, $\eR 1_X = 1_{\ftF X}$.
	On the other hand, suppose that $|X|\ge|Z|$.
	Let~$u,v \in \ftF X$.
	As~$\ftF$ is~$\zeta$-bounded, there are set~$A$ of cardinality $|A|\le|Z|$ and an injective map~$i \colon A \rightarrowtail X$ such that $u,v\in\img(\ftF i)$.
	Then, since~$i$ is injective, we have~$1_A = i^\circ \cdot 1_X \cdot i$.
	Hence, since $|A|\le|Z|$, by~\Cref{p:2f} and naturality of~$\eR$, $1_{\ftF A} = \eR 1_A = (\ftF i)^\circ \cdot \eR 1_X \cdot \ftF i$.
	Therefore, $u\, \eR 1_X\, v$ iff~$u = v$.
\end{proof}

\begin{example}\label{example:z-simulations}
  Let~$M=(M,+,0)$ be a commutative monoid. Then the
  \emph{monoid-valued functor}~$M^{(-)}$ is defined on sets~$X$ by
 $M^{(X)}$ being the set of maps~$\mu\colon X\to M$ with \emph{finite
    support}, i.e.\ with~$\mu(x)=0$ for all but finitely many~$x$. On
  maps~$f\colon X\to Y$, $M^{(-)}$ is defined by
 $M^{(f)}(\mu)(y)=\sum_{f(x)=y}\mu(x)$ for~$\mu\in M^{(X)}$ and
 $y\in Y$. Coalgebras for~$M^{(-)}$ are \emph{$M$-weighted transition
    systems}, i.e.\ finitely branching transition systems in which
  every transition is labelled with an element
  of~$M$~\cite{GummSchroder01}. It has been shown recently \cite[Corollary 4.23]{GoncharovHofmannEtAl25} that
 $M^{(-)}$ admits a normal lax extension iff~$M$ is \emph{positive},
  i.e.\ if~$m+n=0$ implies~$m=n=0$ for~$m,n\in M$. For non-positive
 $M$, for instance for~$M$ being the additive group of the integers,
  we obtain by \cref{p:6,p:10d} that there is no notion of simulation
  on $M$-weighted transition systems that is closed under composition
  and contains all bounded maps  (i.e.\ morphisms of
  $M^{(-)}$-coalgebras) and their converses, and is sound and complete for behavioural
  equivalence.\lsnote{We could relate
    this example to a positive one where we apply the coBarr relator.}
\end{example}

\section{The Greatest Normal Lax Extension}
\label{sec:greatest-lax-extension}

We proceed to capitalize on the fact that a given set functor potentially comes with a whole range of relevant relators with associated sound and complete notions of simulation.
In the main result of this section, we show that set functors that preserve inverse images
admit a greatest normal lax extension w.r.t.\ the pointwise order on
relators.  As we have seen in \Cref{sec:sim-ind-rel}, this means essentially that
for these functors, there is a maximally permissive notion of simulation such that similarity is sound and complete and whose class of simulations contains all coalgebra homomorphisms, their converses and is closed under composition.
Before tackling the main result of this section, we note that it is easy to establish when the greatest normal relational connector or the greatest difunctionally functorial relator exist. 

\begin{proposition}
	A set functor admits a difunctionally functorial relator iff it preserves 1/4-iso pullbacks.
	Furthermore, the coBarr relator of a set functor that preserves 1/4-iso pullbacks is its greatest difunctionally functorial relator.
\end{proposition}
\begin{proof}
	If a set functor admits a difunctionally functorial relator, then it is well-defined on difunctional relations, and it is known that this condition is equivalent to preserving 1/4-iso pullbacks \cite[Theorem 3.12]{GoncharovHofmannEtAl25}. 
	On the other hand, if a set functor preservers 1/4-iso pullbacks, then it has a difunctionally functorial relator: its coBarr relator, which it is clearly the greatest one.
\end{proof}

\begin{proposition}
	The pointwise supremum of a non-empty family of relational connectors of $\ftF$ yields a relational connector of $\ftF$.
\end{proposition}
\begin{proof}
	Let $(\eR_i)_{i \in \calI}$ be a non-empty family of relational connectors of $\ftF$ and let $\eR^\vee$ denote the corresponding pointwise supremum relator. 	
	Moreover, let $r \colon X \relto Y$ be a relation, and let $f \colon A \to X$ and $g \colon B \to Y$ be functions.
	Then, as relational composition preserves supremums, 
	$\eR^\vee (g^\circ \cdot r \cdot f) = \bigvee_{i \in \calI} \eR_i(g^\circ \cdot r \cdot f)
	                                    = \bigvee_{i \in \calI} \big((\ftF g)^\circ \cdot \eR_i r \cdot \ftF f \big)
	                                    = (\ftF g)^\circ \cdot \big(\bigvee_{i \in \calI} \eR_i r\big) \cdot \ftF f 
	                                    = (\ftF g)^\circ \cdot \eR^\vee r \cdot \ftF f$.
	Therefore, as for every set $X$, $1_X \leq \eR^\vee 1_X$ since $\calI$ is non-empty, $\eR^\vee$ is a relational connector. 
\end{proof}

\begin{corollary}
	Every set functor that admits a normal relational connector admits a greatest one.
\end{corollary}

The situation for lax extensions is far more complicated because the supremum of lax extensions is not given by ther pointwise supremum.
We begin by describing the supremum of lax extensions with the help of
relax extensions.

The pointwise infimum of a family of lax extensions of $\ftF$ is again a lax extension of~$\ftF$ and,
therefore, the partially ordered class of lax extensions of~$\ftF$ is
complete, with infimum given by pointwise infimum.  Since the subclass
of normal lax extensions is downwards closed, it follows that~$\ftF$
admits a greatest normal lax extension iff the supremum of every
family of normal lax extensions of~$\ftF$ is normal.
However, while the non-empty pointwise supremum of relax extensions is a relax
extension, as the next example shows, in general the non-empty supremum of lax
extensions does not coincide with the pointwise supremum.

\begin{example}
  Consider the `upper' and `lower' lax extensions
  \(\eF^{\Box} \colon \REL \to \REL\) and
  \(\eF^{\Diamond} \colon \REL \to \REL\), respectively, of the
  powerset functor \(\ftP \colon \SET \to \SET\) that are defined on
  relations \(r \colon X \relto Y\) and sets \(A \subseteq X\) and
  \(B \subseteq Y\) by
  \begin{align*}
    A \mathrel{(\eF^\Box r)} B \Leftrightarrow \forall a \in A. \exists b \in B.\ a\mathrel{r}b; \\*
    A \mathrel{(\eF^\Diamond r)} B \Leftrightarrow \forall b \in B. \exists a \in A.\ a\mathrel{r}b.
  \end{align*}
  The relax extension \(\eR^{\vee} \colon \REL \to \REL\) given by the
  pointwise supremum of \(\eF^\Box\) and \(\eF^\Diamond\) does not
  preserve composition laxly and, hence, is not a lax extension.  For
  instance, over \(2 = \{0,1\}\), we have
  \(\{0\} \mathrel{\eF^\Box 1_2} \{0,1\}\) and
  \(\{0,1\} \mathrel{\eF^\Diamond 1_2} \{1\}\), hence
  \(\{0\} (\eR^{\vee} 1_2 \cdot\eR^{\vee} 1_2) \{1\}\); but
  \(\eR^{\vee} 1_2\) does not relate \(\{0\}\) to \(\{1\}\).
\end{example}

Indeed, the supremum of lax extensions is given by the laxification~\cite{GoncharovHofmannEtAl25} of their pointwise supremum.
The \df{laxification} \(\laxif{\eR} \colon \REL \to \REL\) of a relax extension \(\eR \colon \REL \to \REL\) of \(\ftF\) is the lax extension of~$\ftF$ defined on \(r \colon X \relto Y\) by
\begin{equation}
	\label{p:410}
		\laxif{\eR} r = \bigvee_{\substack{r_1, \ldots, r_n \colon \\ r_n \cdot \ldots \cdot r_1 \leq r}} \eR r_n \cdot \ldots \cdot \eR r_1.
\end{equation}
\begin{proposition}
	\label{cor:flat_sup}
	Let \((\eF_i)_{i \in \calI}\) be a family of lax extensions of \(\ftF\), and let~$\eF^\vee$ be the relax extension given by pointwise supremum of the~\(\eF_i\).
	Then the laxification of\/~$\eF^\vee$  is the supremum of the~\(\eF_i\) in the ordered class of lax extensions.
\end{proposition}
\begin{proof}
	Immediate from the fact that sending a relax extension to its laxification defines a left adjoint~\cite[Proposition 4.1]{GoncharovHofmannEtAl25}.
\end{proof}

Our strategy to show that suprema of normal lax extensions
of~$\ftF$ are normal relies on~$\ftF$ having the property that every
normal lax extension of~$\ftF$ preserves composition with
subidentities.  This allows
restricting the collection of composable sequences of relations needed in~\eqref{p:410} to produce laxifications.
We first fix some notation that we need throughout this section.

Given a relation~$r \colon X \relto Y$, we denote the inclusions~$\dom(r) \rightarrowtail X$ and~$\cod(r) \rightarrowtail Y$ by~$\dor_r$ and~$\cor_r$, respectively.
Furthermore, for an injective map~$i\colon A\ito X$, we denote by~$\subid i$ the induced subidentity
relation~$i\cdot i^\circ\colon X\relto X$. In particular, for a relation \(r
\colon X \relto Y\) we have subidentities  \(\subid {\dor_r} \leq 1_X\) and \(\subid{\cor_r} \leq 1_Y\).

\begin{lemma}
  \label{p:2}
  Let\/ \(\eR\) be a relax extension of\/ \(\ftF\) that preserves composition
  with subidentities. Suppose that\/
  \(\eR r_n \cdot \ldots \cdot \eR r_1 \leq 1_{\ftF X}\) for every
  set~$X$ and every composable sequence \(r_1, \ldots, r_n\) of
  relations such that \(r_n \cdot \ldots \cdot r_1 \leq 1_X\) and
  \(\cod(r_{i-1}) = \dom(r_i)\) for \(i=2,\ldots,n\). Then the
  laxification of\/ \(\eR\) is normal.
\end{lemma}
\begin{proof}
Let~$X$ be a set, and let \(r_1, \ldots, r_n\) be a composable
sequence of relations such that \(r_n \cdots r_1 \leq 1_X\).
By~\eqref{p:410}, to show that the laxification of~$\eR$ is normal
we need to show that
\(\eR r_n \cdot \ldots \cdot \eR r_1 \leq 1_{\ftF X}\) .

By~\cite[Lemma A.1]{GoncharovHofmannEtAl25}, we obtain  composable sequences of  relations
\begin{enumerate}
\item \(s_1,\ldots,s_n\) defined by \(s_n = r_n\), and  \(s_i = \subid{\dor_{s_{i+1}}} \cdot r_i\), for \(i = 1,\ldots,n-1\);
\item \(t_1,\ldots,t_n\) defined by \(t_1 = s_1\), and \(t_i = s_i \cdot \subid{\cor_{t_{i-1}}}\), for \(i = 2,\ldots, n\),
\end{enumerate}
that satisfy \(t_n \cdots t_1 = s_n \cdots s_1 = r_n \cdots r_1\) and  \(\cod(t_{i-1}) = \dom(t_i)\) for  \(i=2,\ldots,n\).
Therefore, since \(\eR\) preserves composition with subidentities and \(\eR t_n \cdot \ldots \cdot \eR t_1 \leq 1_{\ftF X}\) by hypothesis,
\begin{flalign*}
&&\eR r_n \cdots \eR_1 r_1 &\; = \eR s_n \cdot \eR \subid{{\dor}_{s_n}} \cdot \eR r_{n-1} \cdots \eR r_1  &\\*
&&&\; = \eR s_n \cdot \eR (\subid{{\dor}_{s_n}} \cdot r_{n-1}) \cdots \eR r_1  &\\
&&&\; = \eR s_n \cdot \eR s_{n-1} \cdot \eR \subid{{\dor}_{s_{n-1}}} \cdot \eR r_{n-2} \cdots \eR r_1  &\\
&&&\;\quad\vdots\\
&&&\; = \eR s_n \cdots \eR s_1 \\
&&&\; = \eR s_n \cdots \eR s_2 \cdot \eR \subid{{\cor}_{t_1}} \cdot \eR t_1  &\\
&&&\; = \eR s_n \cdot \ldots \eR (s_2 \cdot \subid{{\cor}_{t_1}}) \cdot \eR t_1  &\\
&&&\; = \eR s_n \cdots \eR s_3 \cdot \eR \subid{{\cor}_{t_2}} \cdot \eR t_2 \cdot \eR t_1  &\\
&&&\;\quad\vdots\\
&&&\;= \eR t_n \cdots \eR t_1 \leq 1_{\ftF X}.\tag*{\qed}
\end{flalign*}
\noqed\end{proof}

\begin{remark}
	If a relax extension satisfies the condition of \Cref{p:2} then it is necessarily normal.
\end{remark}

To apply \Cref{p:2} to relax extensions given as pointwise suprema of
normal lax extensions, we next show that preservation of composition
with subidentities is stable under suprema.

\begin{proposition}
  \label{p:19e}
  The relax extension given by the
  pointwise supremum of a non-empty family of normal lax
  extensions\/~$\eF_i$, $i\in I$, of\/~$\ftF$ preserves composition with
  subidentities iff all~$\eF_i$ preserve composition with subidentities.
\end{proposition}

It turns out that the functors whose normal lax extensions preserve composition with subidentities are essentially the ones that preserve inverse images.
To see this, first we record some useful properties of lax extensions:

\begin{lemma}
	\label{p:10c}
	Let~$\eF$ be a lax extension of~$\ftF$.
	Then:
	\begin{enumerate}
		\item \label{p:1c} For every relation~$r \colon X \relto Y$, $\img(\ftF \dor_r) \subseteq \dom(\eF r)$.
		\item \label{p:2c} If~$\eF$ is normal, then for every difunctional relation~$r \colon X \relto Y$, $\eF(r^\circ) = (\eF r)^\circ$.
		\item \label{p:47} If~$\eF$ is normal and~$\ftF$ preserves empty intersections, then for all  injective maps \(i \colon A \ito X\) and~$j \colon A \rightarrowtail Y$, 
	\[
		\eF(j \cdot i^\circ) = \ftF j \cdot (\ftF i)^\circ.
	\]
	\end{enumerate}
\end{lemma}

In the next result we show that the normal lax extensions of a functor that preserve inverse images are precisely the ones that when applied to a relation~$r$ return a relation whose (co)domain is completely determined by the action of the functor on the (co)domain of~$r$.

\begin{lemma}
	\label{prop:lax_rstr}
	
	The following clauses are equivalent:
	\begin{tfae}
		\item \label{p:43a} \(\ftF\) preserves inverse images.
		\item \label{p:44b} $\ftF$ admits a normal lax extension and, for every normal lax extension \(\eF\) of \(\ftF\) and every difunctional relation \(r \colon X \relto Y\), the map~
	\[
		\ftF \catfont{d}_r \colon \ftF \dom(r) \rightarrowtail \ftF X
	\]
	 corestricts to an isomorphism \(\ftF(\dom r) \cong \dom(\eF r)\).
		\item \label{p:44a} $\ftF$ admits a normal lax extension and, for every normal lax extension \(\eF\) of \(\ftF\) and every relation \(r \colon X \relto Y\), the maps
	\[
		\ftF \catfont{d}_r \colon \ftF \dom(r) \rightarrowtail \ftF X \text{ and } \ftF \catfont{i}_r \colon \ftF \cod(r) \rightarrowtail \ftF Y
	\]
	 corestrict to isomorphisms \(\ftF(\dom r) \cong \dom(\eF r)\) and \(\ftF(\cod r)  \cong \cod(\eF r)\), respectively.
		\item \label{p:46c} $\ftF$ admits a normal lax extension and for every normal lax extension~$\eF$ of~$\ftF$, every relation~$r \colon X \relto Y$ and all injective maps~$i \colon X \rightarrowtail A$ and~$j \colon Y \rightarrowtail B$, 
	\[
		\eF(j \cdot r \cdot i^\circ) = \ftF j \cdot \eF r \cdot (\ftF i)^\circ.
	\]
		\item \label{p:47c} $\ftF$ admits a normal lax extension that satisfies the condition of \ref{p:46c}. 
	\end{tfae}
\end{lemma}

Finally, as a consequence of the previous result, we obtain the above mentioned characterization  of the functors that preserve inverse images in terms of their normal lax extensions.

\begin{proposition}
	\label{p:12c}
	The following clauses are equivalent:
	\begin{tfae}
		\item \label{p:4c} $\ftF$ preserves inverse images.
		\item \label{p:7c} $\ftF$ admits a normal lax extension and \emph{each} of its normal lax extensions preserves composition with partial monomorphisms.
		\item \label{p:8c} $\ftF$ admits a normal lax extension that preserves composition with partial monomorphisms.
		\item \label{p:5c} $\ftF$ preserves empty intersections, admits a normal lax extension and \emph{each} of its normal lax extensions preserves composition with subidentities.
		\item \label{p:6c} $\ftF$ preserves empty intersections and admits a normal lax extension that preserves composition with subidentities.
	\end{tfae}
\end{proposition}
\noindent We stress that preservation of empty intersections is a very mild condition that from a coalgebraic point of view is similarly harmless as preservation of monomorphisms.
\begin{proof}
	\ref{p:4c}~$\Rightarrow$ \ref{p:7c}.
	Every functor that preserves inverse images admits a normal lax extension~\cite{GoncharovHofmannEtAl25}, and the fact that each normal lax extension preserves composition with partial monomorphisms is a straightforward consequence of \Cref{prop:lax_rstr}\ref{p:46c}.
	Indeed, let~$\eF \colon \REL \to \REL$ be a normal lax extension of $\ftF \colon \SET \to \SET$.
	Suppose that~$r \colon X \relto Y$ is a relation and~$s \colon A \relto X$ is a partial monomorphism.
	Then, there are injective maps~$i \colon S \to A$ and~$j \colon S \to X$ such that~$s = j \cdot i^\circ$.
	Hence, by \Cref{prop:lax_rstr}\ref{p:46c} and the fact that~$\eF$ is a relational connector, $\eF(r \cdot j \cdot i^\circ) = \eF r \cdot \ftF j\cdot  (\ftF i)^\circ$.
	Therefore, by \Cref{p:10c}(\ref{p:47}), $\eF(r \cdot j \cdot i^\circ) =  \eF r \cdot \eF s$.
	The case of postcomposition with a partial monomorphism follows analogously.

	\ref{p:7c}~$\Rightarrow$ \ref{p:8c}.
	Trivial.

	\ref{p:8c}~$\Rightarrow$ \ref{p:4c}.
	Follows from the fact that a normal lax extension is a normal relational connector and from~\Cref{prop:lax_rstr}\ref{p:47c} since an injective map and the converse of an injective map are particular partial monomorphisms.

	\ref{p:7c}~$\Rightarrow$ \ref{p:5c} and \ref{p:8c}~$\Rightarrow$ \ref{p:6c}.
	We already shown \ref{p:4c} $\Leftrightarrow$ \ref{p:7c} $\Leftrightarrow$ \ref{p:8c}, hence, \ref{p:8c} in particular entails that $\ftF$ preserves intersections.
	\lsnote{@Pedro: How does preservation of empty intersections follow here?}
	Therefore, the claims follow immediately since every subidentity is a partial monomorphism. 

	\ref{p:5c}~$\Rightarrow$ \ref{p:6c}.
	Trivial.

	\ref{p:6c}~$\Rightarrow$ \ref{p:4c}.
	Let~$\eF \colon \REL \to \REL$ be a normal lax extension of~$\ftF$ that preserves composition with subidentities.
	We will see that for every relation~$r \colon X \relto Y$ and all injective maps~$i \colon X \rightarrowtail A$ and~$j \colon Y \rightarrowtail B$, $\eF(j \cdot r \cdot i^\circ) = \ftF j \cdot \eF r \cdot (\ftF i)^\circ$.
	Then, the claim follows by \Cref{prop:lax_rstr}.
	Let~$r \colon X \relto Y$ be a relation and~$i \colon X \rightarrowtail A$ be an injective map.
	Then~$i \cdot i^\circ$ is a subidentity and as~$\ftF$ preserves empty intersections, by \Cref{p:10c}(\ref{p:47}), $\eF(i \cdot i^\circ) = \ftF i \cdot (\ftF i)^\circ$.
	Therefore, by hypothesis, $\eF(r \cdot i^\circ) = \eF(r \cdot i^\circ \cdot i \cdot i^\circ) = \eF(r \cdot i^\circ) \cdot \eF(i \cdot i^\circ) = \eF(r \cdot i^\circ) \cdot \ftF i \cdot (\ftF i)^\circ = \eF r \cdot (\ftF i)^\circ$.
	Similarly, we obtain that for every injective map~\mbox{$j \colon Y \rightarrowtail B$}, $\eF(j \cdot r) = \ftF j \cdot \eF r$, and the claim follows.
\end{proof}

Now we are ready to show the main result of this section:

\begin{theorem}
	\label{thm:main}
	Every set functor that preserves inverse images admits a greatest normal lax extension.
\end{theorem}
\begin{proof}
  Let \(\ftF \colon \SET \to \SET\) be a functor that preserves
  inverse images.  Then, by~\cite{GoncharovHofmannEtAl25}, \(\ftF\) has a normal lax
  extension and, therefore, the least lax extension of~\(\ftF\)  is
  normal.  To show that the non-empty supremum of normal lax
  extensions is normal, due to \Cref{p:12c} and \Cref{p:19e}, we use the criterion of
  \Cref{p:2}.  Let~$(\eF_i)_{i \in \calI}$ be a non-empty family of
  normal lax extensions of~$\ftF$ and let~$\eF^\vee$ be the
  corresponding relax extension given by pointwise supremum.  We have
  to show that for every composable sequence of relations
 $r_1, \ldots, r_n$ such that~$r_n \cdot \ldots \cdot r_1 \leq 1_X$,
  for some set~$X$, and~$\cod(r_{i-1}) = \dom(r_i)$, for
 $i=2, \ldots, n$,
 $\eF^\vee r_n \cdot \ldots \cdot \eF^\vee r_1 \leq 1_{\ftF X}$.
  First, we note that we can assume w.l.o.g. that all $r_i$ are total and surjective, which entails~$r_n \cdot \ldots \cdot r_1 = 1_X$.
  Indeed, %
 by
  (co)restricting each relation in the sequence to its (co)domain we
  obtain a composable sequence~$r'_1, \ldots, r_n'$ of total and
  surjective relations such that~$r'_n \cdot \ldots \cdot r'_1 = 1_A$,
  where~$A = \dom(r_1)$.
  Furthermore, being the pointwise supremum of lax extension that satisfy the condition in \Cref{prop:lax_rstr}\ref{p:46c} as $\ftF$ preserves inverse images,
  $\eF^\vee$ also satisfies this condition. 
  Hence, as $\ftF$ preserves monomorphisms, we obtain
 $\eF^\vee r_n \cdot \ldots \cdot \eF^\vee r_1 = \ftF i \cdot \eF^\vee
  r'_n \cdot \ldots \cdot \eF^\vee r'_1 \cdot (\ftF i)^\circ$, where
 $i \colon A \rightarrowtail X$ is the inclusion of~$A$ into~$X$.
  Therefore, as~$\ftF$ preserves monomorphisms,
 $\eF^\vee r_n \cdot \ldots \cdot \eF^\vee r_1 \leq 1_{\ftF X}
  \Leftrightarrow \eF^\vee r'_n \cdot \ldots \cdot \eF^\vee r'_1 \leq 1_{\ftF
    A}$.
	Now, we proceed by induction on~$n$.
	The base case~$n = 1$ is trivial as~$\eF^\vee$ is normal.
	For the inductive step from~$n$ to~$n+1$, let~$r_1, \ldots, r_{n+1}$ be a composable sequence of total and surjective relations such that~$r_{n+1} \cdot \ldots \cdot r_1 = 1_X$.
	Then, by~\cite[Lemma~4.20]{GoncharovHofmannEtAl25}, $\hat{r}_{n+1} \cdot \ldots \cdot r_1 = 1_X$, where~$\hat{r}$ is the difunctional closure of~$r_{n+1}$.
	Furthermore, as normal lax extensions coincide on difunctional relations, relational composition preserves suprema, and lax extensions preserve composition laxly,
	$\eF^\vee \hat{r}_{n+1} \cdot \eF^\vee r_n =  \eF^\vee \hat{r}_{n+1} \cdot (\bigvee_{i\in \calI} \eF_i r_n) = \bigvee_{i\in \calI} (\eF_i \hat{r}_{r+1} \cdot \eF_i r_n) \leq \bigvee_{i\in \calI} \eF_i (\hat{r}_{r+1} \cdot r_n) = \eF^\vee(\hat{r}_{n+1} \cdot r_n)$.
	Therefore, by induction hypothesis,
	$\eF^\vee r_{n+1} \cdot \ldots \cdot \eF^\vee r_1 \leq \eF^\vee \hat{r}_{n+1} \cdot \ldots \cdot \eF^\vee r_1 \leq \eF^\vee (\hat{r}_{n+1} \cdot r_n)  \cdot \ldots \cdot \eF r_1 \leq 1_{\ftF X}$.
\end{proof}

Recall from \cref{{example:z-simulations}} that a monoid-valued
functor admits a normal lax extension iff the monoid is
positive~\cite{GoncharovHofmannEtAl25}, which in turn is known to be equivalent
to preservation of inverse images by the monoid-valued
functor~\cite{GummSchroder01}. Therefore,

\begin{corollary}
	A monoid-valued functor admits a greatest normal lax extension iff the monoid is positive.
\end{corollary}

\section{Case Study: Labelled Transitions}
\label{sec:lts}

The first and so far only example reported in the literature of a functor that
admits more than one normal lax extension, due to Paul Levy, involves a
monoid-valued functor for a fairly sophisticated submonoid of the non-negative
reals generated as a division semiring by a transcendental
number~\cite[Example~4.11]{DorschMiliusEtAl18}.
We will show that functors of the form~$C + B \times \ftId$ admit a unique normal lax extension.
However, we will also show that it is not uncommon for a functor to have multiple normal lax extensions; in fact, we give a simple and widely used class of examples: Almost all exponential functors, i.e.\ functors of the form~$(-)^A$ (which in coalgebra represent deterministic~$A$-labelled transitions), admit multiple normal lax extensions.
For brevity, we write~$\ftH_A$ for~$(-)^A$.
Building on the results of the previous sections, we describe the complete lattice of normal lax extensions of~$\ftH_A$, in particular obtaining a maximally permissive sound  and complete notion of (bi)simulation for deterministic automata whose class of (bi)simulations is closed under composition. By combining this result with the usual notion of bisimulation on~$\ftP$-coalgebras (i.e.\ unlabelled transition systems), we obtain a new notion of \emph{twisted bisimulation} on labelled transition systems that is more permissive than standard Park-Milner bisimulations.

Since the functor \(\ftH_A\) preserves limits, it preserves weak pullbacks, and its Barr extension sends a relation \(r \colon X \relto Y\) to the relation \(\overline{\ftH_A} r \colon \ftH_A X \relto \ftH_A Y\)
defined by:
\begin{align*}
	f~\overline{\ftH_A} r~g \Leftrightarrow \forall a \in A, f(a)~r~g(a) \Leftrightarrow 1_A \le g^\circ \cdot r \cdot f.
\end{align*}
 We obtain non-standard notions of
simulation by additionally allowing other relations on~$A$ in place of~$1_A$ in the last inequality: We work with a set \(\calA\) of endorelations on \(A\) and
define
\begin{equation}\label{eq:submonoid-lax}
	f~\widehat{\ftH}_A^\calA r~g \Leftrightarrow (\exists \phi \in \calA.\quad \phi \leq g^\circ \cdot r \cdot f),
\end{equation}
or, in pointful notation, $\exists \phi \in \calA,\: \forall a,b\in A, a\:\phi\: b\Rightarrow f(a)\:r\: g(b)$.

We show next  that this construction yields a lax extension of~$\ftH_A$ whenever~$\calA$
forms a submonoid of~$\REL(A,A)$, and that normality of this lax extension is
captured as a condition on the relations in~$\calA$.

\begin{definition}
	An endorelation~$\phi\colon A \relto A$ is \df{normal} if its difunctional closure is reflexive, i.e., if for every \(a\) in \(A\) there is a chain \(a~r~x_1~r^\circ~x_2~r~ \dots~ r^\circ x_{n-1}~r~a\) of alternating~$r$- and~$r^\circ$-steps of (necessarily odd) length~$n\ge 1$.
\end{definition}

\begin{proposition}\label{prop:normal-rel}
	A relation \(\phi \colon A \relto A\) is normal iff for every set~$X$ and every pair of functions~$f,g\colon A\to X$, $\phi\le g^\circ \cdot f$ implies~$f=g$.
\end{proposition}
\begin{proof}
	Let~$\phi \colon A \relto A$ be a normal relation, and~$f,g \colon A \to X$ be functions such that~$\phi\le g^\circ \cdot f$.
	Then, by definition of difunctional closure, $\hat{\phi} \leq g^\circ \cdot f$, where~$\hat{\phi}$ denotes the difunctional closure of~$\phi$.
	Hence, by normality of~$\phi$, $1_A \leq g^\circ \cdot f$ which is equivalent to~$g \leq f$.
	Therefore, as~$f$ and~$g$ are functions, $f = g$.
	To see the converse statement, suppose that the difunctional closure~$\hat{\phi}$ of~$\phi$ is given by~$g^\circ \cdot f$.
	Then, by definition of difunctional closure, $\phi \leq g^\circ \cdot f$.
	Hence, by hypothesis, $f = g$.
	Therefore, $\hat{\phi} = g^\circ \cdot g \geq 1_A$.
\end{proof}
Intuitively, normal relations correspond to roundabout ways of proving
equality. For instance, if~$A=\{a,b\}$, then two functions~$f,g\colon A\to X$
are of course equal if we can show that~$f(a)=g(a)$ and~$f(b)=g(b)$, but also
if we instead show that~$f(a)=g(b)$, $g(a)=g(b)$ and~$g(a)=f(b)$. This second
proof corresponds to the relation~$\{(a,b),(b,b),(b,a)\}$ from~\Cref{fig:bis}.

\begin{theorem}
	\label{p:7}
	Let \(\calA\) be a submonoid of the monoid of endorelations on a set \(A\).
	Then assigning to every relation \(r \colon X \relto Y\) the relation \(\widehat{\ftH}_A^\calA r \colon \ftH_A X \relto \ftH_A Y\) defines a lax extension of \(\ftH_A\) to \(\REL\).
	Furthermore,
	\begin{enumerate}
		\item \label{p:58} if \(\calA\) is closed under converses, then \(\widehat{\ftH}_A^\calA\) preserves converses, and
		\item \label{p:59} if every relation in \(\calA\) is normal, then \(\widehat{\ftH}_A^\calA\) is normal.
	\end{enumerate}
\end{theorem}
\begin{proof}
	Let \(r \colon X \relto Y\) and \(s \colon Y \relto Z\) be relations, and \(f,f' \colon A \to X\), \(g \colon A \to Y\), \(h \colon A \to Z\) and~$t\colon X\to Y$ be functions.
	\begin{itemize}[wide]
		\item[\emph{Monotonicity}.] Trivial.
		\item[\emph{Lax preservation of composition}.] Suppose that \(f~\widehat{\ftH}_A^\calA r~g\) and \(g~\widehat{\ftH}_A^\calA s~h\).
		      Then there exist \(\phi,\psi \in \calA\) such that \(\phi \leq g^\circ \cdot r \cdot f\) and \(\psi \leq h^\circ \cdot s \cdot g\).
		      Hence,
		      \begin{align*}
			      \psi \cdot\phi \leq h^\circ \cdot s \cdot g \cdot g^\circ \cdot r \cdot f \leq h^\circ \cdot s \cdot r \cdot f,
		      \end{align*}
		      and the claim follows from the fact that \(\calA\) is closed under composition.
		\item[\emph{Extension of functions}.]
		      We have~$1_A \le (t \cdot f)^\circ \cdot t \cdot f = f^\circ \cdot t^\circ
			      \cdot (t\cdot f)$, and thus, since~$1_A \in \calA$, $f~\widehat{\ftH}_A^\calA
			      t~(t\cdot f)$ and~$(t\cdot f)~\widehat{\ftH}_A^\calA t~f$.
	\end{itemize}
	When \(\calA\) is closed under converses, we show \ref{p:58}. calculating as
	follows:
	\begin{align*}
		g~\widehat{\ftH}_Ar^\circ~f \Leftrightarrow \phi \leq f^\circ \cdot r^\circ \cdot g %
		\Leftrightarrow \phi^\circ \leq g^\circ \cdot r \cdot f %
		\Leftrightarrow f~\widehat{\ftH}_Ar~g.
	\end{align*}%
	Moreover, by definition, \(f' (\widehat{\ftH}_A^\calA 1_X) f\) iff there is
	\(\phi \in \calA\) such that \(\phi \leq f^\circ \cdot f'\). Hence, if every
	relation in \(\calA\) is normal we obtain \(1_X \leq f^\circ \cdot f'\).
	Therefore, \(f = f'\), which yields~\ref{p:59}.
\end{proof}
Conversely, every lax extension \(\eF \colon \REL \to \REL\) of \(\ftH_A\)
gives rise to a set~$\ftS(\eF)$ of endorelations given by
\begin{equation}\label{eq:lax-submonoid}
	\ftS(\eF) = \{ \phi \colon A \relto A \mid 1_A\: \eF \phi\: 1_A \}
\end{equation}

\begin{proposition}
	\label{p:16}
	Let \(\eF \colon \REL \to \REL\) be a lax extension of the functor \(\ftH_A \colon \SET \to \SET\).
	Then the set\/ \(\ftS(\eF)\) is an upwards-closed submonoid of \(\REL(A,A)\).
	Furthermore, \(\ftS(\eF)\) is closed under converses if\/~\(\eF\) preserves converses, and every relation in \(\ftS(\eF)\) is normal if\/~\(\eF\) is normal.
\end{proposition}
\begin{proof}
	All claims are straightforward except maybe the last. So
	suppose that \(\eF\) is normal.  To show that
	\(r \in \ftS(\eF)\) is normal, we use \Cref{prop:normal-rel}. So let~$r\le d$ where~$d$ is
	of the form~$d=g^\circ \cdot f$ where
	\(f,g \colon A \to X\) are functions.  Then, as \(\eF\) is a normal relational connector, \(\eF r \leq \eF d = (\ftH_A g)^\circ \cdot\ftH_A f\).  Hence, by definition of \(\ftS(\eF)\),
	\((\ftH_A g)^\circ \cdot\ftH_A f\) relates \(1_A\) to \(1_A\),
	which by definition of \(\ftH_A\) entails \(f = g\).
\end{proof}
These two constructions are inverses of each other:

\begin{theorem}
	\label{p:15}
	Every lax extension\/ \(\eF \colon \REL \to \REL\) of the functor\/ \(\ftH_A \colon \SET \to \SET\) is induced by the set\/ \(\ftS(\eF) = \{ r \colon A \relto A \mid 1_A\,(\eF r)\, 1_A \}\) and every upwards-closed submonoid~$\calA$ of~$\REL(A,A)$ is induced by the lax extension~$\widehat{\ftH}_A^\calA$.
\end{theorem}
\begin{proof}
	Let \(\eF \colon \REL \to \REL\) be a lax extension of the functor \(\ftH_A \colon \SET \to \SET\), \(r \colon X \relto Y\)  a relation, and let \(f \colon A \to X\) and \(g \colon A \to Y\) be functions.
	First note that if there is \(\phi \in \ftS(\eF)\) such that \(\phi \leq g^\circ\cdot r \cdot f\), then, by definition of \(\ftS(\eF)\) and local monotonicity of \(\eF\), \(\eF(g^\circ \cdot r \cdot f)\) relates \(1_A\) to \(1_A\).
	Therefore, the claim follows from the fact that \(\eF r\) relates \(f=\ftH_A f(1_A)\) to \(g=\ftH_A g(1_A)\) iff \(\eF(g^\circ \cdot r \cdot f)\) relates \(1_A\) to \(1_A\).

	Conversely, let~$\calA$ be an upwards-closed submonoid of~$\REL(A,A)$. Then
	\begin{equation*}
		\phi \in \ftS(\widehat{\ftH}_A^\calA)
		\Leftrightarrow 1_A\,\widehat{\ftH}_A^\calA\phi\,1_A
		\Leftrightarrow \exists \psi\in\calA,\;\psi\le\phi
		\Leftrightarrow \phi\in\calA. \qedhere
	\end{equation*}
\end{proof}

\begin{corollary}
	\label{p:2e}
	The normal lax extensions of the functor \(\ftH_A \colon \SET \to \SET\) correspond precisely to the upward-closed submonoids of \(\REL(A,A)\) consisting only of normal relations.
\end{corollary}

\begin{example}
	\label{p:600}
	\begin{enumerate}[wide]
		\item \label{p:601} With \(A = 2 = \{a,b\}\), consider the upwards-closed submonoid~$\calA$ of~$\REL(2,2)$ generated from the single relation \(\Phi = \{(a,b),(b,b),(b,a)\}\).
		      The lax extension~$\widehat{\ftH}_2^\calA$ of \(\ftH_2 \colon \SET \to \SET\) to \(\REL\) is normal, preserves converses and differs from the Barr extension of \(\ftH_2\) since \(\widehat{\ftH}_2^\calA\Phi\) relates \(1_2\) to itself but \(\overline{\ftH}_2\Phi\) does not.
		      The corresponding notion of bisimulation, combined with the standard notion of bisimulation for unlabelled transition systems, is the `twisted' bisimulation mentioned in the introduction.
		\item \label{p:20e} With \(A = 3 = \{a,b,c\}\), consider the upwards-closed submonoid $\calA$ of
		      $\REL(3,3)$ generated from the single relation \(\Phi = \ftP(3 \times
		      3)\setminus\{(a,a),(b,c)\}\). We obtain a normal lax extension $\eF$ of \(\ftH_3
			\colon \SET \to \SET\) that does not preserve converses and, hence, differs
		      from the Barr extension. Indeed, \(\widehat{\ftH}_3^\calA\Phi\) relates \(1_3\)
		      to itself but \(\widehat{\ftH}_3^\calA (\Phi^\circ)\) does not.
			As far as we know, this is the first example of a non-symmetric normal lax extension.
			Furthermore, since~$\eF$ is normal, by~\Cref{p:1}, $\eF$-similarity is sound and complete.
	\end{enumerate}
\end{example}

\Cref{p:15} implies in particular that the class~$\catfont{Lax}(\ftH_A)$ of lax extensions of~$\ftH_A$ is small, so we will regard it as a set. It is easy to see that the mutually inverse constructions described in \Cref{p:7} and \Cref{p:16} define monotone maps \(\ftS \colon \catfont{Lax}(\ftH_A) \to \catfont{SubMon}^\uparrow(\REL(A,A))\) and \(\widehat{\ftH}_A^{(-)} \colon \catfont{SubMon}^\uparrow(\REL(A,A)) \to \catfont{Lax}(\ftH_A)\) between the partially ordered set \(\catfont{Lax}(\ftH_A)\) of lax extensions of the functor \(\ftH_A \colon \SET \to \SET\) and the partially ordered set \(\catfont{SubMon}^\uparrow(\REL(A,A))\) of upwards-closed submonoids of \(\REL(A,A)\) ordered by inclusion.
This allows reasoning about suprema of lax extensions in terms of suprema of
sets of endorelations in \(\catfont{SubMon}^\uparrow(\REL(A,A))\) which is
given by closure under composition of union of the sets; more specifically, the
supremum of a family \((\calA)_{i \in \calI}\) of elements of
\(\catfont{SubMon}^\uparrow(\REL(A,A))\), denoted as \(\bigvee_{i\in
	\calI}\calA_i\), is the smallest set that contains
\(\bigcup_{i\in\calI}\calA_i\) and is closed under composition.

\begin{proposition}
	\label{p:32}
	Let \((\eF_i)_{i \in \calI}\) be a family of lax extensions of the functor \(\ftH_A \colon \SET \to \SET\).
	Then the supremum \(\bigvee_{i \in \calI} \eF_i\) is given by the lax extension induced by set \(\bigvee_{i\in \calI}\ftS(\eF_i)\).
\end{proposition}
\begin{example}
	\label{p:7e}
	The squaring functor \(\ftH_2 \colon \SET \to \SET\) has precisely four normal lax extensions, which are induced by the upwards-closed submonoids of \(\REL(2,2)\)
	\begin{itemize}
		\item \(\calA_\bot\) generated by the empty set -- the Barr extension;
		\item \(\calA_a\) generated from the single relation \(\Phi_a = \{(a,b),(a,a),(b,a)\}\);
		\item \(\calA_b\) generated from the single relation \(\Phi_b = \{(a,b),(b,b),(b,a)\}\);
		\item \(\calA_\top\) generated by the set \(\{\Phi_a,\Phi_b\}\).
	\end{itemize}
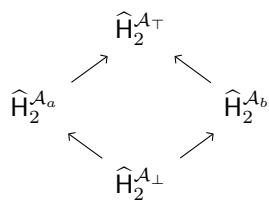
	\begin{center}
		\begin{tikzcd}[column sep=small, row sep=small]
			& \widehat{\ftH}_2^{\calA_\top}  & \\
			\widehat{\ftH}_2^{\calA_a}  &                                & \widehat{\ftH}_2^{\calA_b} \\
			& \widehat{\ftH}_2^{\calA_\bot}  &
			\ar[from=3-2, to=2-1]
			\ar[from=3-2, to=2-3]
			\ar[from=2-1, to=1-2]
			\ar[from=2-3, to=1-2]
		\end{tikzcd}
	\end{center}
\end{example}

Since exponential functors preserve terminal objects, all states belonging to coalgebras for a exponential functor are behaviourally equivalent.
This means that similarity w.r.t. the greatest lax extension of a exponential functor -- which induces the most permissive notion of simulation but that in general fails to be normal -- is sound, complete and the corresponding class of simulations is closed under composition.
The situation is far more interesting for polynomial functors, whose definition we recall next.

Given a family~$\calF = (\ftF_s)_{s \in S}$ of  set functors indexed by a set~$S$, we denote by~$\Sigma(\calF) \colon \SET \to \SET$ the canonical functor ``sum of a family of functors'' that sends a set~$X$ to the coproduct~$\sum_{s \in S} \ftF_s X$, and we denote by~$c^s \colon \ftF_s \to \Sigma(\calF)$ the natural transformation whose~$X$-component is defined by the coprojection~$c^s_X \colon \ftF_s X \to \sigma(\calF)$.

\begin{definition}
	A polynomial set functor is a functor of the form~$\Sigma(\calH)$ for a family~$\calH$ of exponential functors.
\end{definition}

Recall from \Cref{p:14e} that almost all polynomial functors (except
the exponential functors) are~$\zeta$-bounded, so that \Cref{p:10d}
applies. In particular, this means that notions of simulation on
coalgebras for polynomial functors induced by a lax extension are
sound iff the lax extension is normal. This further motivates the
investigation of the structure of the lattice of normal lax extensions
of polynomial functors. We show next that greatest normal lax
extensions of polynomial functors are constructed from greatest normal
lax extensions of exponential functors; we will illustrate this
principle on the minimization of deterministic automata, which are
coalgebras for polynomial functors of the form~$2\times\ftH_A$.

In the remainder of the paper, we fix a set~$S$ and a family~$\calF = (\ftF_s)_{s \in S}$ of set functors, and we say that a family $\calL=(\eF_s)_{s\in S}$ of lax extensions  is a family of lax extensions of~$\calF$ if for every~$s \in S$, $\eF_s$ is a lax extension of~$\ftF_s$.

It is easy to see that a family~$\calL$ of lax extensions of~$\calF$ gives rise to a lax extension $\Sigma^\eF(\calL)$ of~$\Sigma(\calF)$  given by $\Sigma^\eF(\calL) r=\bigvee_{s \in S} \big( c^s_Y \cdot \eF_s r \cdot (c^s_X)^\circ \big)$ for $r \colon X \relto Y$;
or, in pointful notation, for all~$u \in \Sigma(\calF) X$ and~$v \in \Sigma(\calF) Y$, $u \mathrel{\Sigma^\eF(\calL) r} v$ iff there is~$s \in S$ and~$u' \in \ftF_s X$, $v' \in \ftF_s Y$ s.t.~$u = c^s_X(u'), v = c^s_Y(v')$ and~$u' \mathrel{\eF_s r} v'$.
Then, assigning to every family~$\calL$ of lax extensions of~$\calF$ the lax extension~$\Sigma^\eF(\calL)$ defines a monotone map~$\Sigma^\eF(-) \colon \catfont{Fam}(\calF) \to \catfont{Lax}(\Sigma{\calF})$ from the partially ordered class $\catfont{Fam}(\calF)$ of families of lax extensions of~$\calF$, ordered pointwise, to the partially ordered class $\catfont{Lax}(\Sigma{\calF})$ of lax extensions of~$\Sigma(\calF)$ ordered pointwise.

Conversely, every lax extension~$\eF$ of the functor~$\Sigma(\calF)$ gives rise to an~$S$-indexed family~$c^\ast(\eF)$ of lax extensions of~$\calF$ by ``(co)restricting'' the action of the lax extension to~$\ftF_s$, i.e., for every~$s \in S$ and every relation~$r \colon X \relto Y$, $c^\ast(\eF)_s (r)$ is given by~${(c^s_Y)}^\circ \cdot \eF r \cdot c^s_X$;
or in pointful notation, for all~$u' \in \ftF_s X$ and~$v' \in \ftF_s Y$, $u' \mathrel{c^\ast(\eF)_s (r)} v'$ iff~$c^s_X(u') \mathrel{\eF r} c^s_Y(v')$.
Then, assigning to every lax extension~$\eF$ of~$\Sigma(\calF)$ the family~${c}^\ast(\eF)$ of lax extensions of~$\calF$ defines a monotone map~${c}^\ast(-) \colon \catfont{Lax}(\Sigma{\calF}) \to \catfont{Fam}(\calF)$.

Immediately from the definitions we have:

\begin{proposition}
	\label{p:11e}
	\pnnote{@Pedro: add proof if time permits}
	The map~$\Sigma^\eF(-)$ is an order reflecting left adjoint of~$c^\ast(-)$.
\end{proposition}

Of course, in general, the constructions defined above are not inverse of each other.

\begin{example}
	Consider the pair~$(\ftC_1,\ftC_1)$ where~$\ftC_1 \colon \SET \to \SET$ denotes the constant functor to~$1=\{\ast\}$.
	The functor~$\ftC_1 + \ftC_1$ is isomorphic to the constant functor~$\ftC_2 \colon \SET \to \SET$ to~$2=\{0,1\}$.
	The greatest lax extension~$\eF^\top$ of~$\ftC_2$ sends every relation to the greatest relation on~$2$ which is different from the identity on~$2$.
	However, $c^\ast(\eF^\top) = (\overline{\ftC_1},\overline{\ftC_1})$ and~$\Sigma^\eF(c^\ast(\eF^\top))$ is the Barr extension of~$\ftC_2$ which sends every relation to the identity map on~$2$.
\end{example}

\noindent However, as we show next, they become inverse of each other when~$\calF$ consists of functors that weakly preserve pullbacks and only normal lax extensions are allowed.

Let~$\catfont{NLax}(\Sigma{\calF})$ be the partially ordered subclass of $\catfont{Lax}(\Sigma{\calF})$ given by the normal lax extensions of~$\Sigma{\calF}$, and let~$\catfont{NFam}(\calF)$ be the partially ordered subclass of
$\catfont{Fam}(\calF)$ given by the families of normal lax extensions of~$\calF$.
Clearly, the map~$\Sigma^\eF(-)$ (co)restricts to a map~$\Sigma^\eF(-) \colon \catfont{NFam}(\calF) \to \catfont{NLax}(\Sigma{\calF})$, and, as each natural transformation~${c^s}$ is monic, it is easy to see that the map~${c}^\ast(-)$ (co)restricts to~${c}^\ast(-) \colon \catfont{NLax}(\Sigma{\calF}) \to \catfont{NFam}(\calF)$.

\begin{remark}
	\label{p:10e}
        It has  been observed that the canonical forgetful functor from the category of lax extensions to the category of \(\SET\)-endofunctors is topological~\cite{SchubertSeal08}, and hence, in particular, has a left adjoint.
         This left adjoint picks, for every \(\ftF \colon\SET\to\SET\), the smallest element of the fibre \(\catfont{Lax}(\ftF)\) of \(\ftF\) with respect to this forgetful functor.
        We also note that, if \(\ftF\) weakly preserves pullbacks, this element is given by the Barr extension of \(\ftF\).

	Suppose now that every functor in the family~$\calF$ weakly preserves pullbacks.
	Let~$\calL = (\ftbF_s)_{s \in S}$ be the family of the corresponding Barr extensions.
	It is easy to see that the functor~$\Sigma(\calF)$ weakly preserves pullbacks as well, and, by adjointness, its Barr extension is given by~$\Sigma^\eF(\calL)$.
\end{remark}

\begin{theorem}
	\label{p:4e}
	\pnnote{This also works for relational connectors}
	If all functors in~$\calF$ weakly preserve pullbacks, then the maps~$\Sigma^\eF(-) \colon \catfont{NFam}(\calF) \to \catfont{NLax}(\Sigma{\calF})$ and~${c}^\ast(-) \colon \catfont{NLax}(\Sigma{\calF}) \to \catfont{NFam}(\calF)$ are inverse of each other.
\end{theorem}
\begin{proof}
	By \Cref{p:11e}, it remains only to show that $\Sigma^\eF({c}^\ast(\eF)) \geq \eF$.
	Let~$r \colon X \relto Y$, and let~$u \in \Sigma(\calF) X$ and~$v \in \Sigma(\calF) Y$ s.t.~$u \mathrel{\eF r} v$.
	Suppose that all functors in~$\calF$ weakly preserve pullbacks.
	Then, by \Cref{p:10e}, the functor~$\Sigma(\calF)$ weakly preserves pullbacks.
	Hence, as~$\eF$ is a normal relational connector, by~\Cref{p:5}, $u \mathrel{\overline{\Sigma(\calF)} (\hat{r})} v$, where~$\hat{r}$ is the difunctional closure of~$r$.
	Thus, by \Cref{p:10e}, there is~$s \in S$, $u' \in \ftF_s X$ and~$v' \in \ftF_s Y$ s.t.~$u = c^s_X(u')$ and~$v = c^s_Y(v')$.
	Hence, by definition, $u' (\mathrel{{c}^\ast(\eF) r)} v'$.
	Therefore, by definition of sum of a family of lax extensions, $u \mathrel{\Sigma^\eF({c}^\ast(\eF)) r} v$.
\end{proof}

Thererefore:

\begin{corollary}
	\label{p:1z}
	The greatest normal lax extension of a polynomial set functor determined by a family~$\calH$ of exponential functors is given by the sum of the family of the greatest normal lax extensions of the exponential functors in~$\calH$.
\end{corollary}

Interestingly, as a consequence of our results, we obtain that the usual notion of (bi)simulation for stream systems with termination is the only sound notion induced by a normal lax extension:

\begin{corollary}
	\label{p:5e}
	For all sets~$C$ and~$B$, the functor~$C + B \times \ftId \colon \SET \to \SET$ admits a unique normal lax extension.
\end{corollary}
\begin{proof}
	From~\Cref{p:2e} with $A = \varnothing$ and $A = 1$ we conclude that the constant functor $1$ and the identity functor admit a unique normal lax extension, respectively.
	Now, the claim follows from~\Cref{p:4e}.
\end{proof}

In particular, the maximally permissive notion of simulation induced by a lax extension for deterministic automata can be obtained as follows.

\begin{corollary}
	Let~$A,B$ and~$C$ be sets.
	The greatest normal lax extension of~$C+B\times \ftH_A$ is given by postcomposing the Barr extension of~$C+B\times \ftId$ with the greatest normal lax extension of~$\ftH_A$.
\end{corollary}
\begin{proof}
	Straightforward consequence of \Cref{p:1z}, \Cref{p:5e} and the fact, mentioned in \Cref{p:10e}, that whenever every functor in the family $\calF$ weakly preserves pullbacks, $\overline{\Sigma\calF} = \Sigma^\eF(\calL)$, where $\calL$ is the the family of lax extensions of  $\calF$ consisting of the corresponding Barr extensions.
\end{proof}

\begin{example}[Automata]\label{expl:minimization}
	Let us reinterpret~\Cref{p:7e}  in automata-theoretic terms.
	Let~$S = \{s_0,\ldots,s_{n-1}\}\cup \{t_0,\ldots,t_{m-1}\}$ be the state space
	of a deterministic automaton with all states accepting, with~$n$ and~$m$ being
	mutually prime, and with~$\{a,b\}$ being the alphabet of actions. The
	transition function is as follows: the~$a$-transitions connect every~$s_i$ with
	$s_{(i+1)\!\!\mod n}$ and every~$t_j$ with~$t_{(j+1)\!\!\mod m}$; the
	$b$-transitions connect every~$s_i$ with~$t_0$ and every~$t_j$ with~$s_0$.
	Since all states are accepting, they are all bisimilar, and thus the minimal
	automaton has one state. A bisimulation relation in the usual sense --  which is induced by Barr extension ($\widehat{\ftH}_2^{\calA_\bot}$ in \Cref{p:7e}) -- for showing equivalence of
	any~$s_n$ and any~$t_m$ must include every pair~$(s_i,t_j)$, and thus is
	potentially quadratic in size. The following non-standard bisimulation of
	linear size for the greatest normal lax extension~$\widehat{\ftH}_2^{\calA_\top}$ from~\Cref{p:7e} can be used instead:~$R = S\times
		\{s_0,t_0\}\cup \{s_0,t_0\}\times S$. This is essentially because either~$s_0$
	or~$t_0$ is reachable by a~$b$-transition from any state in one step, and they
	are related to everything by~$R$.
\end{example}

To conclude, we derive the notion of twisted bisimulation on
labelled transition systems (LTS) mentioned in the introduction as follows. For
the sake of readability, we continue to restrict to the case where the set of
labels is~$2=\{a,b\}$. Then~$2$-labelled transition systems are coalgebras for
the functor~$\ftF=\ftH_2\cdot\ftP$ where~$\ftP$ is the covariant powerset
functor.
For~$\calA$ being one of the upwards closed submonoids of
\(\REL(2,2)\) listed in \Cref{p:7e}, we obtain a normal lax
extension~$\widehat{\ftF}^\calA$ of~$\ftF$ by composing
$\widehat{\ftH}_2^{\calA}$ with~$\ftbP$ (cf.~\Cref{p:18d}), i.e.\
\begin{equation*}
	\widehat{\ftF}^\calA r=\widehat{\ftH}_2^{\calA}(\ftbP r)\qquad\text{for~$r\colon X\relto Y$}.
\end{equation*}
Then, the most permissive notion of twisted bisimulation is the one
induced by~$\widehat{\ftF}^{\calA_\top}$. In terms of the standard
representation of~$2$-labelled LTS as pairs~$(X,(\to_l)_{l\in 2})$
consisting of a set~$X$ of states and transition relations
$\to_l\subseteq X\times X$ (always denoted by the same symbol if no
confusion is likely), this notion is explicitly described as follows:
Given LTS~$(X,(\to_u)_{u\in 2})$, $(Y,(\to_u)_{u\in 2})$, a relation
$r\colon X\relto Y$ is a \emph{twisted bisimulation} (for~$\calA_\top$)
if whenever~$x\mathrel{r}y$, then one of the following clauses holds
(cf.\ \Cref{fig:bis}):
\begin{enumerate}
	\item Whenever~$x\to_u x'$, then there exists~$y\to_uy'$ such that~$x'\mathrel{r}y'$,
	      and whenever~$y\to_uy'$, then there exists~$x\to_ux'$ such that
	     $x'\mathrel{r}y'$, for~$u\in 2$.
	\item For~$(u,v)\in\{(a,b),(b,a),(b,b)\}$, whenever~$x\to_u x'$, then there exists
	     $y\to_v y'$ such that~$x'\mathrel{r}y'$, and whenever~$y\to_vy'$, then there
	      exists~$x\to_u x'$ such that~$x'\mathrel{r}y'$.
	\item Dito, for~$(u,v)\in\{(a,b),(b,a),(a,a)\}$.
\end{enumerate}
(In particular, every bisimulation in the standard sense is a twisted
bisimulation.)  Since~$\widehat{\ftF}^{\calA_\top}$ is a normal lax
extension, twisted bisimulation is sound (and complete) for standard
bisimilarity, i.e.\ two states are bisimilar if (and only if) they are
related by a twisted
bisimulation~\cite[Theorem~11]{MartiVenema15}. Our statement from the
introduction to the effect that twisted bisimulations on LTS can be
smaller than standard bisimulations is illustrated by
\Cref{expl:minimization}, as a deterministic automaton with all states
accepting is in particular an LTS.  
Generally speaking, twisted bisimulations can have advantages when the successors of a state under different labels are bisimilar.
\pnnote{the successors of states under different labels like $(p,q)$ under $a,b$ in the example?}
This is illustrated in the following elementary example.

\begin{example}
	Consider the following LTS.
	
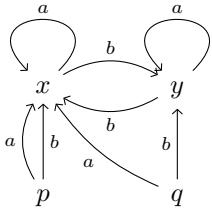
\begin{center}
		\begin{tikzcd}[row sep=large, column sep=large]
            		x & y \\
                  	p & q
                        \ar[from=1-1, to=1-1, "a"', loop]
                        \ar[from=1-1, to=1-2, "b", bend left]
                        \ar[from=1-2, to=1-2, "a"', loop]
                        \ar[from=1-2, to=1-1, "b", bend left]
                        \ar[from=2-1, to=1-1, "a", bend left]
                        \ar[from=2-1, to=1-1, "b"']
                        \ar[from=2-2, to=1-1, "a", bend left=15]
                        \ar[from=2-2, to=1-2, "b"]
                  \end{tikzcd}
	\end{center}
	Then, the smallest bisimulation in the usual sense relating~$p$ and~$q$ is $\{(p,q),(x,x),(x,y),(y,y),(y,x)\}$, while under twisted bisimulation, we can make do with the strictly smaller relation $\{(p,q),(x,x),(x,y),(y,y)\}$.
\end{example}
Since the functor
$\ftH_A \cdot \ftP$ preserves weak pullbacks, due to
\Cref{p:5}(\ref{p:8e}) and \Cref{p:5}(\ref{p:10}), another way of
thinking about twisted bisimulations is that they form a subclass of
bisimulations (in the usual sense) up to difunctionality that, unlike
the full class of bisimulations up to difunctionality, is closed under
relational composition.
This phenomenon generalizes as follows:

\begin{corollary}
      Let $\ftF \colon \SET \to \SET$ be a functor that admits a normal lax extension, and let $\eF_0 \colon \REL \to \REL$ be its least normal lax extension.
      Then, for every normal lax extension $\eF$ of $\ftF$, the class of $\eF$-simulations is contained in the class of $\eF_0$-simulations up-to difunctionality. 
\end{corollary}

Therefore, having a nice description of the least normal lax extension, which  we known that is given by the laxification of the Barr relax extension \cite[Corollary 4.2]{GoncharovHofmannEtAl25}, can help us to understand the simulations induced by each normal lax extension.
However, we stress that, as it can be observed in \Cref{p:666}, the class of $\eF$-simulations up to difunctionality usually is \emph{not} closed under composition. 
In contrast, the greatest normal lax extension induces the most permissive notion of simulation that allows reasoning compositionally about behavioural equivalence.
\pnnote{Does this make sense? I don't really like the last sentence}

\section{Conclusions}

\noindent We have analysed aspects of notions of (bi)simulation induced by
relators and lax extensions, reinforcing the view that a given functor
(i.e.\ a given system type) can be associated with multiple relevant
notions of (bi)simulation. By establishing key results on the
existence and properties of lax extensions, we have clarified their
role in certifying behavioural equivalence across diverse system
types. Notably, we have demonstrated that functors preserving 1/4-iso
pullbacks admit a sound and complete notion of  bisimulation induced by
the coBarr relator, and that normality is essentially a necessary
condition for soundness.  Furthermore, we have shown that functors
preserving inverse images possess a greatest normal lax extension,
providing a maximally permissive notion of bisimulation.

In a case study on functors of the form~$(-)^A$, which model
$A$-labelled transitions, we have introduced the notion of twisted
bisimulation and demonstrated its greater permissiveness compared to
standard bisimilarity, while retaining soundness. This result can
potentially offer new tools for reasoning about state-based systems of
various branching types, and allows for smaller bisimulations
certifying behavioural equivalence.

Our work contributes to the general theory of bisimulations by
refining the structural conditions under which sound and complete
bisimulation notions exist. One direction for future investigation is
to identify further sufficient conditions for a functor to admit a
greatest normal lax extension; one candidate condition is weak
preservation of 1/4-iso and 4/4-epi pullbacks, which has recently been
shown to guarantee existence of a normal lax
extension~\cite{GoncharovHofmannEtAl25}.
Such endeavour will require a completly new proof strategy because, as we have shown, preservation of inverse images is crucial for the techniques used in this paper.
A further important open
question is the uniqueness of normal lax extensions and a
characterization of functors whose Barr-bisimilarity is
complete. %
\bibliographystyle{IEEEtran}
\bibliography{mpbbibl}

\clearpage

\onecolumn

\appendix

\subsection{Useful pullbacks}

\begin{lemma}
  \label{p:400}
  Let~$f\colon X\to A$ and~$g\colon Y\to A$.
  Then, the following diagram  is a pullback square
	\begin{center}
		\begin{tikzcd}
			\dom(r) &[3ex] g[Y]  \\
			X       & A.
			\ar[from=1-1, to=1-2, "f|_{\dom(r)}"]
			\ar[from=1-1, to=2-1, tail]
			\ar[from=1-1, to=2-2, phantom, "\lrcorner", very near start]
			\ar[from=1-2, to=2-2, tail, "j"]
			\ar[from=2-1, to=2-2, "f"']
		\end{tikzcd}
	\end{center}
where~$r = g^\circ\cdot f$ and~$j\colon g[Y]\rightarrowtail A$ is the obvious inclusion.
\end{lemma}
\begin{proof}
Noting the image factorization~$Y\to g[Y] \xto{j} A$ of~$g$, we form the
pullback
\begin{equation*}
\begin{tikzcd}[column sep=5em, row sep=normal]
f^\mone[g[Y]]
	\rar["f|_{f^\mone[g[Y]]}"]
	\dar[tail]
	\ar[dr,phantom, "\lrcorner", very near start]
&
g[Y]
	\dar[tail,"j"]\\
X
	\rar["f"']
&
A
\end{tikzcd}
\end{equation*}
We are left to show that~$f^\mone[g[Y]]=\dom(r)$. Indeed, $\dom(r) = \cod(r^\circ) = \cod((g^\circ\cdot f)^\circ) = \cod(f^\circ\cdot g) = f^\mone[g[Y]]$.
\end{proof}

\begin{lemma}[e.g.~\cite{GoncharovHofmannEtAl25}[Lemma 2.2]]
	\label{p:48}
	Let \(r \colon X \relto Y\) be a relation.
	Then the following  are equivalent:
	\begin{tfae}
		\item \label{p:51} \(r\) is difunctional;
		\item \label{p:54} for every span~$X\xfrom{\pi_1} R\xto{\pi_2} Y$
		such that \(r = \pi_2 \cdot \pi_1^\circ\), the pushout square
		\begin{center}
			\begin{tikzcd}
				R & Y \\
				X & O
				\ar[from=1-1, to=1-2, "\pi_2"]\
				\ar[from=1-1, to=2-1, "\pi_1"']
				\ar[from=1-1, to=2-2, phantom, very near end, "\ulcorner"]
				\ar[from=2-1, to=2-2, "p_1"']
				\ar[from=1-2, to=2-2, "p_2"]
			\end{tikzcd}
		\end{center}
		is a weak pullback.
	\end{tfae}
\end{lemma}

\subsection{Proof of  \Cref{p:5}}
	\pnnote{I think we can get away with obvious here, no?}
	The clauses \ref{p:15e},  \ref{p:8e} and the first inequality of~$\ref{p:10}$ follow straightforwardly from the definition of coBarr relator.
	To see the second inequality of \ref{p:10}, let~$\eR$ be a normal relational connector and let~$r \colon X \relto Y$ be a relation.
	Furthermore let $\hat{r} \colon X \relto Y$ be the difunctional closure of~$r~$, and let  $f \colon X \to O$ and~$g \colon Y \to O$ be maps s.t.~$\hat{r} = g^\circ \cdot f$.
	Then, as~$\eR$ is a normal relational connector, $\eR(r) \leq (\ftF g)^\circ \cdot \ftF f = \ftcF r$. 
	On the other hand, let~$p \colon A \to X$ and~$q \colon A \to Y$ be maps s.t.~$r = q \cdot p^\circ$.
	Then, $1_A \leq q^\circ \cdot r \cdot p$.
	Hence, as~$\eR$ is a normal relational connector, $1_{\ftF A} \leq (\ftF q)^\circ \cdot \eF(r) \cdot \ftF p$.
	Therefore, by adjointness, $\ftF q \cdot (\ftF p)^\circ \leq \eR r~$.
	\qed

\subsection{Details for  \Cref{p:14e}}

We prove the claim that a polynomial functor is~$\zeta$-bounded if it
is not an exponential functor.

Let~$\Sigma(\calH) \colon \SET \to \SET$ be a polynomial functor, with
$\calH = (\ftH_{A_s})_{s \in S}$, for some set~$S$. If~$S=\emptyset$,
then~$\Sigma(\calH)$ is a constant functor, hence~$\zeta$-bounded. If
$|S|=1$, then~$\Sigma(\calH)$ is an exponential functor.  Thus,
suppose that~$|S| \geq 2$. Moreover, if all~$A_s$ are empty, then
$\Sigma(\calH)$ is a constant functor, hence~$\zeta$-bounded; so
suppose that there is~$s \in S$ such that~$A_s \neq \varnothing$. Let
$Z$ be the carrier of a terminal coalgebra for~$\Sigma(\calH)$, which
is well-known to exist (e.g.~\cite{Rutten00}). Now for all~$s \in S$
and~$X \in \SET$, $|1 + \ftH_{A_s} (X)| \leq |\Sigma(\calH)(X)|$.
Thus, for every~$s \in S$, the cardinality of~$Z$ is greater than or equal
to the cardinality of the terminal coalgebra for~$1 + \ftH_{A_s}$,
which is greater than or equal to the cardinality of~$A_s$ and it is
infinite for~$A_s$ non-empty (cf.~\cite[Example 10.2(6)]{Rutten00}).
In particular, this means that the cardinality of~$Z$ is infinite.
Now, let~$u, v \in \Sigma(\calH)(X)$ for some set~$X$.  We claim that
there is a set~$Y$ and an injective map~$i \colon Y \rightarrowtail X$
s.t.~$u,v \in \img(\Sigma(\calH)i)$ and~$|Y| \leq |Z|$.  By definition
of~$\Sigma(\calH)$ there are~$s,t \in S$ and~$u' \in \ftH_{A_s}(X)$,
$v' \in \ftH_{A_s}(X)$ s.t~$u = c^s_X(u')$ and~$v = c^t_X(v')$.  Let
$Y$ be the set~$\img(u') \cup \img(v')$, and let
$i \colon Y \rightarrowtail X$ be the inclusion of~$Y$ into~$X$.
Then, as~$|Z|$ is infinite and~$|\img(u')| \leq |A_s| \leq |Z|$ and
$|\img(v')| \leq |A_t| \leq |Z|$, $|Y| \leq |Z|$.  Furthermore, as
$u'$ and~$v'$ corestrict to~$Y$ and
$c^s \colon \ftH_{A_s}\to \Sigma(\calH)$ and
$c^t \colon \ftH_{A_j} \to \Sigma(\calH)$ are natural transformations,
it follows that~$u,v \in \img(\Sigma(\calH)i)$.
	\pnnote{probably better to add a remark about accessibility
          for infinite |Z|}
\qed
\subsection{Proof of \Cref{p:19e}}

Let~$\eF$ be the pointwise supremum of the~$\eF_i$. Let
 $r \colon X \relto Y$ be a relation, and let~$s \colon X \relto X$
  be a subidentity.  Since normal lax extensions coincide on
  difunctional relations (e.g.~\cite{MartiVenema15,HofmannSealEtAl14}), in particular on subidentities,
  all~$\eF_i$ and, hence~$\eF$ map~$s$ to the same relation~$\bar{s}$,
  which by normality of the~$\eF_i$ is a subidentity.

  \emph{`If:'} Since relational composition preserves suprema,
 $\eF (r \cdot s) = \bigvee_{i \in \calI} \eF_i(r \cdot s) =
  \bigvee_{i \in \calI} (\eF_i r \cdot \bar{s}) = (\bigvee_{i \in
    \calI} \eF_i r) \cdot \bar{s} =\eF r \cdot \bar{s}=\eF r \cdot \eF
  s$.  The case of postcomposition with subidentities runs
  analogously.

  \emph{`Only if:'} We show that
 $\eF_i(r \cdot s) \leq \eF_i r \cdot \eF_i s=\eF_i r \cdot \bar{s}$;
  the other inequality holds by the definition of lax extension.  So
  let~$a \in \ftF X$ and~$b \in \ftF Y$ such that
 $a \mathrel{\eF_i(r \cdot s)} b$.  We claim that
 $a \mathrel{\eF_i r} b$ and~$a \mathrel{\bar{s}} a$.  Indeed, as~$s$
  is a subidentity, $r \cdot s \leq r$, whence
 $a \mathrel{\eF_i r} b$.
  Moreover, $a \mathrel{\eF (r \cdot s)} b$ because~$\eF_i$ is
  below~$\eF$. Since~$\eF$ preserves composition with subidentities,
  we thus have~$c$ such that~$a \mathrel{\bar{s}} c$ and
 $c\mathrel{Lr}b$. But then~$a=c$ because~$\bar{s}$ is a subidentity.
  The case of postcomposition with a subidentity runs
  analogously. \qedhere
\qed

\subsection{Proof of \Cref{p:10c}}

	\begin{enumerate}[wide]
		\item Let~$r \colon X \relto Y$ be a relation.
				Consider a span
	      			\begin{tikzcd}[column sep=small]
		      			X & R & Y
		      			\ar[from=1-2, to=1-1, "\pi_1"']
		      			\ar[from=1-2, to=1-3, "\pi_2"]
	      			\end{tikzcd}
				such that~$r = \pi_2 \cdot \pi_1^\circ$.
				Then, $\pi_1$ factors as~$e \cdot \dor_r$, with~$e \colon R \twoheadrightarrow \dom(r)$ and~$\dor_r \colon \dom(r) \rightarrowtail X$.
				Hence, \(\ftF \pi_2 \cdot (\ftF e)^\circ \cdot (\ftF \dor_r)^\circ \leq \eF r\).
				Therefore, as every set functor preserves epimorphisms, $\img(\ftF \dor_r) \subseteq \dom(\eF r)$.

		\item Let~$X\xto{f} Y\xfrom{g} B$ be a cospan in~$\SET$.
				Then, as~$\eF$ is a normal relational connector, $\eF((g^\circ \cdot f)^\circ) = \eF(f^\circ \cdot g) = (\ftF f)^\circ \cdot \ftF g = (\eF(g^\circ \cdot f))^\circ$.
		\item Let \(i \colon A \ito X\) and~$j \colon A \rightarrowtail Y$ be injective maps.
				Then, \(j \cdot i^\circ \colon X \relto Y\) is difunctional.
				Hence, by \Cref{p:48}, the pushout square of \(i\) along~$j$
				\begin{center}
					\begin{tikzcd}
						A & X \\
						X & O
						\ar[from=1-1, to=1-2, tail, "i"]
						\ar[from=1-1, to=2-1, tail, "i"']
						\ar[from=1-2, to=2-2, tail, "p_2"]
						\ar[from=2-1, to=2-2, tail, "p_1"']
						\ar[from=2-2, to=1-1, phantom, "\ulcorner", very near start]
					\end{tikzcd}
				\end{center}
				is a pullback square, in fact it is the square of an intersection.
				Since~$\ftF$ preserves empty intersections by hypothesis, then it preserves all intersections~\cite[Proposition~2.1]{Trnkova69}.
				Thus, as  \(\eF\) is a normal relational connector, we obtain \(\eF(j \cdot i^\circ) = \eF(p_2^\circ \cdot p_1) = (\ftF p_2)^\circ \cdot \ftF
p_1 = \ftF j \cdot (\ftF i)^\circ\).
	\end{enumerate}

\subsection{Proof of \Cref{prop:lax_rstr}}

	\ref{p:43a} \(\Rightarrow\) \ref{p:44b}.
	By~\cite{GoncharovHofmannEtAl25}, \(\ftF\) has a normal lax extension~$\eF$.
	For~$r$ difunctional, we have \(r=g^\circ \cdot f \) for some cospan~$X\xto{f} A \xfrom{g} Y$.
	Then, by \Cref{p:400} we have a commutative diagram
	\begin{center}
		\begin{tikzcd}
			\dom(r) &[3ex] g[Y] & Y \\
			X       & A,
			\ar[from=1-1, to=1-2, "f|_{\dom(r)}"]
			\ar[from=1-1, to=2-1, tail, "\dor_r"']
			\ar[from=1-1, to=2-2, phantom, "\lrcorner", very near start]
			\ar[from=1-2, to=2-2, tail, "j"]
			\ar[from=1-3, to=1-2, "e"', twoheadrightarrow]
			\ar[from=1-3, to=2-2, "g"]
			\ar[from=2-1, to=2-2, "f"']
		\end{tikzcd}
	\end{center}
	in which the square is a pullback and~$g=j\cdot e$ is the image
	factorization of~$g$.
	Since~$\ftF$ preserves inverse images, it preserves monomorphisms, and every set functor preserves epimorphisms, so~$\ftF j \cdot \ftF e$ is an epi-mono factorization of~$\ftF g$.
	This means that~$\ftF j \colon \ftF(g[Y]) \rightarrowtail \ftF A$ corestricts to an isomorphism~$h \colon \ftF(g[Y]) \cong \ftF g[FY]$.
	Furthermore, as~$\eF$ is a normal relational connector, $\eF r = (\ftF g)^\circ \cdot \ftF f$.
	Thus, by \Cref{p:400} (applied to $\ftF f$, $\ftF g$, and $\eF r$) and \Cref{p:10c}(\ref{p:1c}) we obtain the following diagram where the outer square is a pullback since~$\ftF$ preserves inverse images by hypothesis.
	\begin{center}
		\begin{tikzcd}
			\ftF \dom(r) & \ftF g[Y] \\
			\dom(\eF r)  & \ftF g[\ftF Y] \\
			\ftF X       & \ftF A.
			\ar[from=1-1, to=2-2, phantom, "\lrcorner", very near start]
			\ar[from=1-1, to=1-2, "\ftF(f_{|\dom(r)})"]
			\ar[from=1-1, to=2-1, dashed]
			\ar[from=1-1, to=3-1, "\ftF \dor_r"', bend right=60]
			\ar[from=1-2, to=2-2, "h"]
			\ar[from=1-2, to=3-2, "\ftF j", bend left=60]
			\ar[from=2-1, to=2-2, "\ftF (f_{|\dom(\eF r)})"]
			\ar[from=2-1, to=3-1, "d_{\eF r}", tail]
			\ar[from=2-1, to=3-2, phantom, "\lrcorner", very near start]
			\ar[from=2-2, to=3-2, "k", tail]
			\ar[from=3-1, to=3-2, "\ftF f"']
		\end{tikzcd}
	\end{center}
	Furthermore, as~$k$ is mono and the bottom square and the triangles commute, the top square also commutes.
	Hence, as~$h$ is an isomorphism, we conclude that~$(\ftF \dor_r, h \cdot \ftF(f_{|\dom(r)}))$ is a pullback of~$(\ftF f, k)$.
	Therefore, as pullbacks are unique, $\ftF \dor_r$ corestricts to an isomorphism between~$\ftF\dom(r)$ and~$\dom(\eF r)$.

	\ref{p:44b} \(\Rightarrow\) \ref{p:44a}.
	 Note that the codomain of a relation is the domain of its converse and by \Cref{p:10c}(\ref{p:2c}) every normal lax extension preserves converses of difunctional relations.
	Therefore, the claim follows from \ref{p:44b} by duality.

	 \ref{p:44a} \(\Rightarrow\) \ref{p:46c}.
	 We show the case of precomposition with the converse of an injective map, the other case follows analogously.
	Let \(\eF \colon \REL \to \REL\) be a normal lax extension of \(\ftF\), and let \(r \colon X \relto Y\) be a relation and~$j \colon X \rightarrow A$ be an injective map.
	It suffices to show that \(\eF(r \cdot i^\circ) \leq \eF r \cdot (\ftF i)^\circ\) since the other inequality holds by definition of lax extension.
	We begin by observing that as~$i \cdot i^\circ = \dor_{i^\circ } \cdot {\dor_{i^\circ}}^\circ$, $\ftF i \cdot (\ftF i)^\circ = \ftF \dor_{i^\circ} \cdot (\ftF \dor_{i^\circ})^\circ$.
	\pnnote{Do we need details here? LS: Cite well-definedness  of the Barr extension?}
	Furthermore, as~$i$ is injective, $\eF r = \eF(r \cdot i^\circ \cdot i) = \eF(r \cdot i^\circ) \cdot \ftF i$.
	Hence, the claim follows once we show
	\(
	\eF(r \cdot i^\circ) \leq
	\eF(r \cdot i^\circ) \cdot \ftF i \cdot (\ftF i)^\circ =
	\eF(r \cdot i^\circ) \cdot \ftF \dor_{i^\circ} \cdot (\ftF \dor_{i^\circ)})^\circ
	\)
which is equivalent to showing~$\dom(\eF(r \cdot i^\circ)) \subseteq \img(\ftF \dor_{i^\circ})$.
	To see the latter note that by hypothesis~$\ftF d_{r \cdot i^\circ} \colon \ftF(\dom(r\cdot i^\circ)) \to \dom(\eF(r \cdot i^\circ))$ corestricts to an isomorphism~$h \colon \ftF(\dom(r\cdot i^\circ)) \cong  \dom(\eF(r \cdot i^\circ))$.
	Furthermore, as~$\dom(r \cdot i^\circ) \subseteq \dom(i^\circ)$, we have~$\dor_{r \cdot i^\circ} = \dor_{i^\circ} \cdot k$, where~$k$ denotes the inclusion~$\dom(r \cdot i^\circ) \rightarrowtail \dom(i^\circ)$.
	Hence, the inclusion~$\dom(\eF (r \cdot i^\circ)) \rightarrowtail \ftF X$ factors as~$\ftF \dor_{i^\circ} \cdot \ftF k \cdot h^{-1}$.
	Therefore, $\dom(\eF(r\cdot i^\circ)) \subseteq \img(\ftF \dor_{i^\circ})$.

	\ref{p:46c}~$\Rightarrow$ \ref{p:47c}.
	Trivial.

	\ref{p:47c}~$\Rightarrow$ \ref{p:43a}.
	Let~$X\xto{f} Y\xfrom{i} B$ be a cospan in~$\SET$ with~$i$ injective and~$\eF \colon \REL \to \REL$ be a normal lax extension of~$\ftF$ that satisfies the condition required in \ref{p:47c}.
	Consider a pullback~$(p_1,p_2)$ of~$(f,i)$.
	Then, since pullbacks preserve monomorphisms, $p_1$ is injective.
	Hence, as~$\eF$ is normal, by hypothesis we obtain~$(\ftF i)^\circ \cdot \ftF f = \eF(i^\circ \cdot f) = \eF(p_2 \cdot p_1^\circ) = \ftF p_2 \cdot (\ftF p_1)^\circ$.
	Therefore, $(\ftF p_1, \ftF p_2)$ is a weak pullback of~$(\ftF f, \ftF i)$ and it is in fact a pullback since~$\ftF p_1$ is injective because~$\ftF$ preserve injective maps given that it admits a normal lax extension~\cite{GoncharovHofmannEtAl25}.
\qed

\subsection{Details of \Cref{p:2}} 

\begin{lemma}
    Let \(r_1, \ldots, r_n\) be a composable sequence of  relations.
    Consider the following composable sequences of  relations:
    \begin{enumerate}
    	\item \(s_1,\ldots,s_n\) defined by \(s_n = r_n\), and  \(s_i = \domr{s_{i+1}} \cdot r_i\), for \(i = 1,\ldots,n-1\);
        \item \(t_1,\ldots,t_n\) defined by \(t_1 = s_1\), and \(t_i = s_i \cdot \codr{t_{i-1}}\), for \(i = 2,\ldots, n\).
    \end{enumerate}
    Then, \(t_n \cdots t_1 = s_n \cdots s_1 = r_n \cdots r_1\) and for every \(i=2,\ldots,n\), \(\cod(t_{i-1}) = \dom(t_i)\).
\end{lemma}
\begin{proof}
	Clearly, \(t_n \cdots t_1 = s_n \cdots s_1 = r_n \cdots r_1\). %
        Moreover, note that \(\dom(t_i) = \dom(s_i) \cap \cod(t_{i-1})\), and \(\cod(t_{i-1}) = s_{i-1}[\cod(t_{i-2})]\) if \(i>2\) and \(\cod(t_{i-1}) = \cod(s_1)\) if \(i=2\).
        Thus, we have $\cod(t_{i-1})\le \cod(s_{i-1})$ for $i=2,\ldots,n$, and since \(\cod(s_{i-1}) \leq \dom(s_i)\) for \(i = 2,\ldots,n\), it follows that \(\dom(t_i) = \dom(s_i) \cap \cod(t_{i-1}) = \cod(t_{i-1})\).
\end{proof}

\end{document}